\documentclass[english,12pt]{article}
\usepackage[dvips,letterpaper]{geometry}
\usepackage{graphicx,psfrag,epsf}
\usepackage{times}
\usepackage{fullpage}
\usepackage{amsmath,amssymb,amsthm,amsfonts,mathrsfs,dsfont,tikz}
\usepackage{epsfig}
\usepackage{subfig}
\usepackage{url}
\usepackage{verbatim}
\usepackage{natbib}
\usepackage{geometry,setspace,scalefnt,bm,float,enumerate,multirow,booktabs,pdflscape,rotating,arydshln,color,mathtools,subfig,subfloat}
\usepackage{indentfirst}
\usepackage{placeins}
\usepackage{caption}
\usepackage{hyperref,proof}

\DeclareFontFamily{U}{mathx}{\hyphenchar\font45}
\DeclareFontShape{U}{mathx}{m}{n}{
<5> <6> <7> <8> <9> <10>
<10.95> <12> <14.4> <17.28> <20.74> <24.88>
mathx10
}{}
\DeclareSymbolFont{mathx}{U}{mathx}{m}{n}
\DeclareFontSubstitution{U}{mathx}{m}{n}
\DeclareMathAccent{\widecheck}{0}{mathx}{"71}
\DeclareMathAccent{\wideparen}{0}{mathx}{"75}

\input cyracc.def

\theoremstyle{definition}
\newtheorem{thm1}{Theorem}
\newtheorem{rem1}{Remark}

\newtheorem{theorem}{Theorem}[section]
\newtheorem{lemma}[theorem]{Lemma}
\newtheorem{corollary}[theorem]{Corollary}
\newtheorem{proposition}[theorem]{Proposition}

\title{A Model for Bimodal Rates and Proportions}
\author{Roberto Vila$^{1}$, Lucas Alfaia$^{1}$, Andr\'e F.B. Menezes$^{2}$\\ Mehmet N. \c{C}ankaya$^{3,4}$
\ and \  Marcelo Bourguignon$^5$\thanks{Corresponding
author: Marcelo Bourguignon. Departamento de Estat\'istica, Universidade Federal do Rio Grande do Norte, Natal, RN, Brazil. Email: m.p.bourguignon@gmail.com.}
\\
{\footnotesize $^{1}$Department of Statistics, Universidade de Bras\'ilia, Bras\'ilia, DF, Brazil}
\\[-0.15cm]
{\footnotesize $^{2}$Department of Statistics, Universidade Estadual de Campinas, Campinas, SP, Brazil}
\\[-0.15cm]
{\footnotesize $^{3}$Department of International Trading and Finance, Faculty of Applied Sciences, U\c{s}ak University, U\c{s}ak, Turkey}
\\[-0.15cm]
{\footnotesize $^{4}$Department of Statistics, Faculty of Art and Sciences, U\c{s}ak University, U\c{s}ak, Turkey}
\\[-0.15cm]
{\footnotesize $^{5}$Department of Statistics, Universidade Federal do Rio Grande do Norte, Natal, RN, Brazil}\\[-0.15cm]
}
\date{}

\begin{document}

\maketitle
\begin{abstract}
\noindent

The beta model is the most important distribution for fitting data with the unit interval. However,
the beta distribution is not suitable to model bimodal unit interval data.
In this paper, we propose a bimodal beta  distribution
constructed by using an approach based on the alpha-skew-normal model. We discuss several properties of this distribution such as bimodality, real moments, entropy measures and identifiability. Furthermore, we propose a new regression model based on the proposed model
and discuss residuals. Estimation is performed by maximum
likelihood. A Monte Carlo experiment is conducted to evaluate the performances of these estimators in finite samples with a discussion of the results.
An application is provided to show the modelling competence of the proposed distribution when the data sets show bimodality.\\

\end{abstract}

\section{Introduction}
\noindent

The need for modeling and analyzing bimodal bounded data, in specials for data on the unit interval, occurs in many fields of real life
such as bioinformatics \citep{ji05}, image classification \citep{ma09}, transaction at a car dealership
\citep{Smithson1} and so on. In such situations, in order to apply probabilistic modeling these phenomena, under a
parametric paradigm, probability distributions limited to $[0, 1]$ are indispensable.
Especially, the unimodal beta model is the most widely model used in the literature
to describe data in the unit interval, especially because of its
flexibility and fruitful properties \citep{jo95}.
However, despite its broad sense applicability in many fields, the beta distribution
is not suitable to model bimodal data on the unit interval.

In general, one uses mixtures of distributions for describing the bimodal data.
For example, \cite{Smithson1} and \cite{Smithson2}
consider finite mixtures of beta regression models to analyze the priming
effects in judgments of imprecise probabilities. However, in general, mixtures
of distributions may suffer from identifiability problems in the parameter estimation; see
\cite{li1,li2}. Thus, new mixture-free models which have the capacity to accommodate
unimodal and bimodal are very important as often real-world data are better modeled by these models. The nature of phenomena can show bimodality due to many reasons such as economical policies, uncertainty of social movement and its effects on the economy \citep{wong,VC20}.

Variations of the beta model can be found in \cite{ferrari2004}, \cite{os08}, \cite{bayes12}, \cite{Hahn21}, among others. However, all the models cited above are not suitable for capturing bimodality.
Recently, probabilistic models for modeling bimodality on the positive real line were discussed by various authors. \cite{olmos17} introduced
recently a bimodal extension of the Birnbaum-Saunders distribution.
\cite{VFSPO20} proposed the bimodal gamma distribution.
\cite{VC20} considered a bimodal Weibull distribution. Despite this, to the best of our
knowledge, a specific parametric model to describe bimodality data observed on the unit interval has never been considered in the literature.

Based on the above discussion and motivated by the presence of bimodality in proportion responses, we develop a model for
double bounded response variables. In particular, we extended the usual
beta distribution using a quadratic transformation technique used to generate bimodal functions \citep{e:10}.
The approach therefore appears to be a new development for the literature.
We discuss several properties of the proposed model such as bimodality, real moments, hazard rate, entropy measures and identifiability.
Furthermore, we study the effects of the explanatory variables on the response variable using a regression model.

In what follows, we list some of the main contribution and advantages of the proposed model.
\begin{itemize}
\item We introduce a new family of distributions that is flexible version of the usual
beta distribution so that it is capable of fitting bimodal as well as unimodal data.
We provide general properties of the proposed model;
\item We propose an extend version of the quadratic transformation technique used to generate bimodal functions;
\item The proposed model allows the boundary values to lie on a smooth unified continuum
along with the rest of the open interval (0, 1), as opposed existing as one or two discontinuities, i.e.,
it does not require boundary values to be either discarded or else treated separately \citep{Hahn21}.
Thus, one of the main motivation of this paper is to contribute with another attractive
regression model for modeling of double bounded response variables.
\end{itemize}

The rest of the article proceeds as follows. In Sections \ref{sect:2} and \ref{sect:3}, we present the new distribution and
derive some of its properties.
Then in Section \ref{sect:4}, we present the main properties of the bimodal Beta, which include entropy measures, stochastic representation and identifiability.
Section \ref{sect:5} presents the bimodal Beta regression model.
Also, the estimation method for the model parameters and diagnostic measures are discussed.
In Section \ref{sect:6}, some numerical results of the estimators and the empirical distribution of the residuals are presented with a discussion of the
results.
A real life application related to the proportion of votes that Jair Bolsonaro received in the second turn of Brazilian elections in 2018  is analyzed in Section \ref{sect:7}.
Section \ref{sect:8} summarizes the main findings of the paper.

\section{The Beta bimodal distribution}
\label{sect:2}
\noindent

In this Section, the bimodal Beta (BBeta) distribution is introduced and its density is derived. Moreover,
some results on the bimodality properties are obtained.
We say that a random variable (r.v.) $X$ has a BBeta distribution with parameter vector
$\boldsymbol{\theta}_\delta=(\alpha,\beta, \rho, \delta)$, $\alpha>0,\beta>0$, $\rho\geqslant 0$ and $\delta\in\mathbb{R}$,
denoted by $X\sim \text{BBeta}(\boldsymbol{\theta}_\delta)$,
if its probability density function (PDF) is given by
\begin{align}\label{beta-density}
    f(x;\boldsymbol{\theta}_\delta)
    =
    \begin{cases}
    \displaystyle
    \frac{\rho+(1-\delta{x})^2}{ Z(\boldsymbol{\theta}_\delta) {B}(\alpha,\beta) } \,
    x^{\alpha-1} \, (1-x)^{\beta-1},
    & 0\leqslant x\leqslant 1
    \\[0,2cm]
    0, & \text{otherwise},
    \end{cases}
\end{align}
where
\begin{align}\label{partition-function}
Z(\boldsymbol{\theta}_\delta)
=
1+\rho
- 2\delta\,{\alpha\over \alpha+\beta}
+
\delta^2\,{\alpha(\alpha+1)\over(\alpha+\beta)(\alpha+\beta+1)}
\end{align}
denotes the normalization constant and ${B}(\alpha,\beta)$ is the beta function.
When $\delta=0$, $\rho$ is simplified in \eqref{beta-density}, and then we obtain the classic beta distribution
with parameter vector $\boldsymbol{\theta}_0=(\alpha,\beta,\rho, 0)\coloneqq (\alpha,\beta)$. The parameters $\alpha$, $\beta$ (which appear as exponents of the r.v.) and $\rho$, control the format of the distribution. The uni- or bimodality is controlled  by the parameter $\delta$. Note that for $\alpha$, $\beta$ and $\delta\neq 0$ fixed, the parameter $\rho$ also controls the uni- or bimodality of the distribution. From Figure \ref{pdfbbeta} we note some different shapes of the BBeta PDF for different combinations of parameters. Figure \ref{pdfbbeta} (a) and (b) represent $L$ shape and its bimodal form and
bell shaped case of Beta distribution, respectively.
\begin{figure}[H]
	\centering
	\subfloat[]{\label{fig:bimodalpdfbbeta}\includegraphics[width=0.45\textwidth]{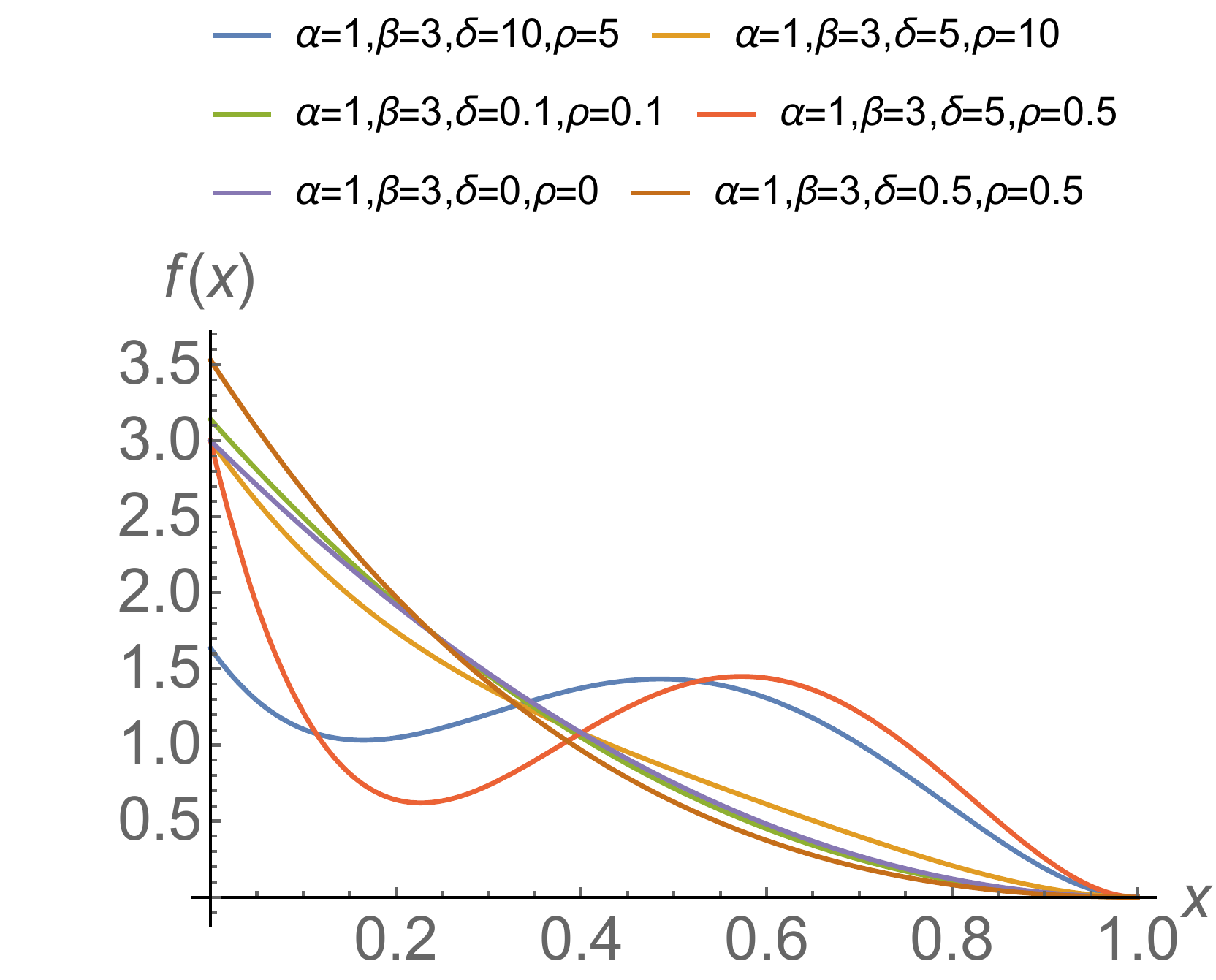}}
	\quad
	\subfloat[]{\label{fig:flucpdfbbeta}\includegraphics[width=0.45\textwidth]{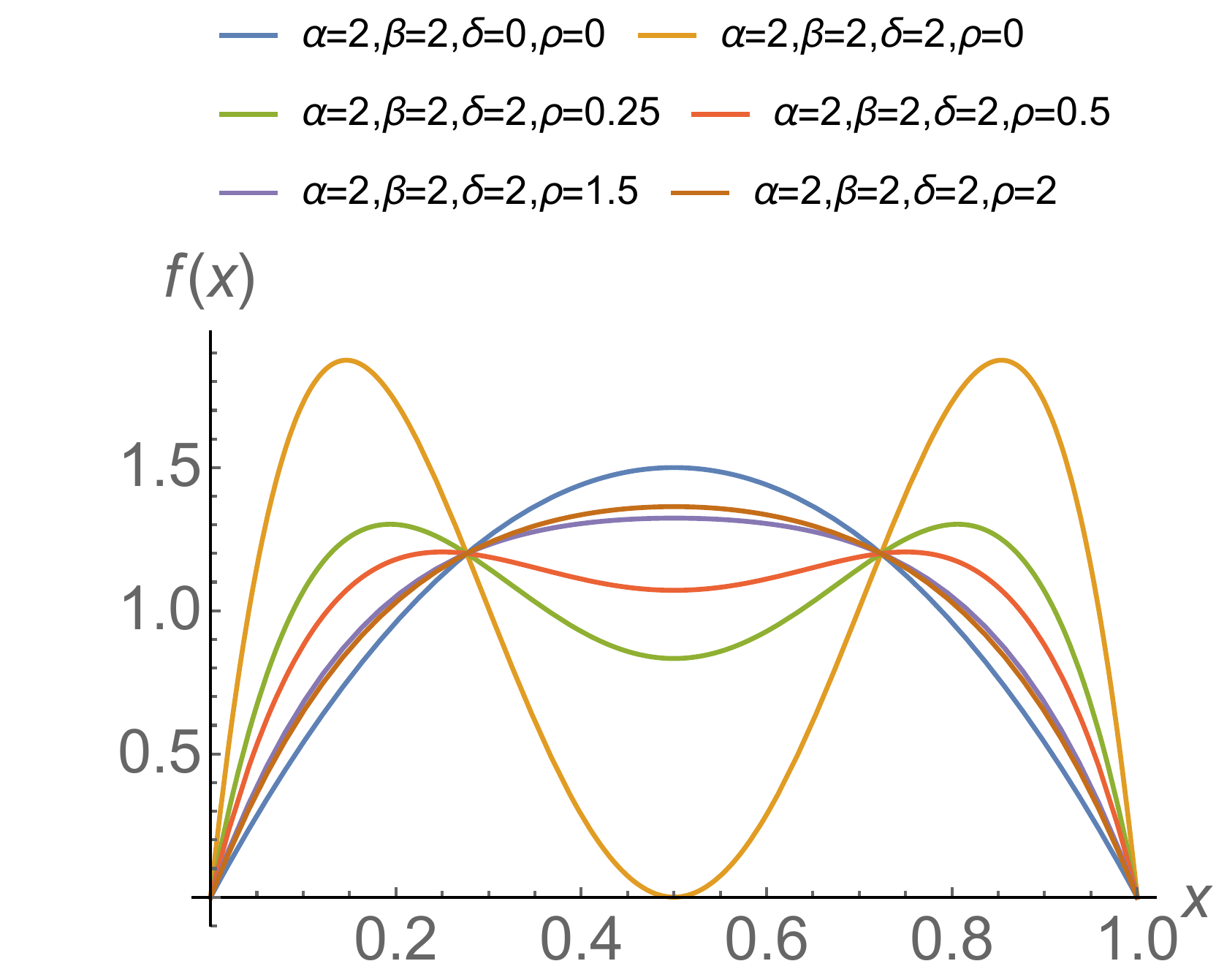}}
	\caption{The PDF of BBeta for different values of parameters.}
		\label{pdfbbeta}
\end{figure}

If $X\sim \text{BBeta}(\boldsymbol{\theta}_\delta)$, the cumulative distribution function (CDF), the survival function (SF) and the hazard rate function (HR) of $X$ are, respectively, given by
\begin{align}
&F(x;\boldsymbol{\theta}_\delta)
=
\dfrac{1}{ Z(\boldsymbol{\theta}_\delta) }
\biggl[
(1+\rho)\,
{I_{x}(\alpha, \beta)}
-2\delta\,
\dfrac{B_{x}(\alpha+1, \beta)}{B(\alpha, \beta)}
+\delta^2\,
\dfrac{B_{x}(\alpha+2, \beta)}{B(\alpha, \beta)}
\biggr],  \label{CDF}
\\[0,2cm]
&S(x;\boldsymbol{\theta}_\delta)
=
\dfrac{1}{ Z(\boldsymbol{\theta}_\delta)  }
\sum_{i=0}^{2}
c_i\,
\biggl[
\dfrac{B(\alpha+i, \beta)}{B(\alpha, \beta)}
-
\dfrac{B_{x}(\alpha+i, \beta)}{B(\alpha, \beta)}
\biggr]
\ \text{and}    \label{SF}
\\[0,2cm]
&H(x;\boldsymbol{\theta}_\delta)
=
\dfrac{\big[\rho+(1-\delta{x})^2\big]  x^{\alpha-1} \, (1-x)^{\beta-1}}{ 		
	\sum_{i=0}^{2} c_i\,
	\big[B(\alpha+i, \beta)
	-
	B_{x}(\alpha+i, \beta)
	\big]},   \label{HR}
\end{align}
where $I_{x}(\alpha, \beta)$ is the incomplete beta function ratio, $B_{x}(\alpha, \beta)$ is the incomplete beta function, and $c_0=1+\rho$, $c_1=-2\delta$, $c_2=\delta^2$. For more details on the derivation of these formulas see Section \ref{sect:3}.

%

\subsection{Bimodality properties}
\noindent

To state the following result that guarantees the bimodality of the \text{BBeta} distribution, we define the set $\mathcal{A}$ formed by all $\boldsymbol{\theta}_\delta=(\alpha,\beta, \rho, \delta)\in(0,+\infty)^2\times[0,+\infty)\times\mathbb{R}$  such that following hold:
\begin{eqnarray}
\alpha>1,\beta>1, \delta> 1,\rho\neq 0; \label{cond-1}
\\[0,1cm]
\delta(\alpha-3)>-2(\alpha+\beta-2); \label{cond-2}
\\[0,1cm]
2\delta(2+\delta-\alpha)<(\rho+1)(\alpha+\beta-2);  \label{cond-3}
\\[0,1cm]
(\rho+1)(\alpha-1)>2\delta.  \label{cond-4}
\end{eqnarray}

Note that the set $\mathcal{A}$ is non-empty because the point $\boldsymbol{\theta}_\delta=(6,6,0.1,2)\in \mathcal{A}$.
\begin{thm1}[Bimodality; case $\rho\neq 0$]
If $X\sim \text{BBeta}(\boldsymbol{\theta}_\delta)$ such that $\boldsymbol{\theta}_\delta\in \mathcal{A}$ then the \text{BBeta} distribution is bimodal.
\end{thm1}
\begin{proof}
A simple computation shows that
\begin{eqnarray}\label{derivative}
    f'(x;\boldsymbol{\theta}_\delta)
    =
    {x^{\alpha-2}(1-x)^{\beta-2}\over Z(\boldsymbol{\theta}_\delta) {B}(\alpha,\beta)} \,
    p_3(x),
\end{eqnarray}
where
\begin{align}\label{polynomial}
p_3(x)= \big[\rho+(1-\delta x)^2\big]\big[(\alpha-1)(1-x)-(\beta-1)x\big]
-2\delta(1-\delta x)x(1-x).
\end{align}
This implies that, $f'(x;\boldsymbol{\theta}_\delta)=0$ if and only if $x=0$, $x=1$ and
\begin{align*}
p_3(x)&=
	-x^3\delta^2(\alpha+\beta-2)+x^2\delta\big[\delta(\alpha-3)+2(\alpha+\beta-2)\big]
	\\[0,2cm]
	&+x\big[2\delta(2+\delta-\alpha)-(\rho+1)(\alpha+\beta-2)\big]
	-2\delta+(\rho+1)(\alpha-1)=0.
\end{align*}
Since, by definition, the boundary points are never critical points we exclude the analysis at these points.
By using \eqref{cond-1} in \eqref{derivative}-\eqref{polynomial} we have $f'(x;\boldsymbol{\theta}_\delta)\neq 0$ for all $x>1$.
In other words, the roots of $p_3(x)$ occur within the interval $(0,1)$.

We claim that, under conditions \eqref{cond-1}, \eqref{cond-2}, \eqref{cond-3} and \eqref{cond-4}, $p_3(x)$ has exactly three different roots within the interval $(0,1)$.

Indeed, under \eqref{cond-1}-\eqref{cond-4}, by Descartes’ rule of signs
(see, e.g. \cite{xue2012loop} and \cite{{griffiths1947introduction}}), $p_3(x)$ has three or one positive roots.
But by conditions \eqref{cond-1}-\eqref{cond-4} and by Vieta’s formula (see, e.g., \cite{vinberg2003course}),
\begin{eqnarray*}
x_1+x_2+x_3={\delta\big[\delta(\alpha-3)+2(\alpha+\beta-2)\big]\over \delta^2(\alpha+\beta-2)},
\\[0,1cm]
x_1x_2+x_2x_3+x_1x_3=-{\big[2\delta(2+\delta-\alpha)-(\rho+1)(\alpha+\beta-2)\big]\over \delta^2(\alpha+\beta-2)},
\\[0,1cm]
x_1x_2x_3={-2\delta+(\rho+1)(\alpha-1)\over \delta^2(\alpha+\beta-2)},
\end{eqnarray*}
we obtain that the polynomial equation $p_3(x)= 0$ has exactly three positive roots $x_1,x_2$ and $x_3$ in $(0,1)$, and the claimed follows.

Without loss of generality, let’s assume that $x_1<x_2<x_3$. Since, for $\alpha>1,\beta>1$, $f(x; \boldsymbol{\theta}_\delta) \longrightarrow 0$ as $x \to 0^+$
and $f(x; \boldsymbol{\theta}_\delta) \longrightarrow 0$  as
$x\to 1^-$, it follows that the BBeta density \eqref{beta-density} increases on the intervals $(0, x_1)$
and $(x_2, x_3)$, and decreases on $(x_1, x_2)$ and $(x_3, 1)$. That is, $x_1$ and $x_3$ are two
maximum points and $x_2$ is the unique minimum point. Thus we have complete the proof of theorem.
\end{proof}

\begin{thm1}[Bimodality; case $\rho= 0$]
	If $X\sim \text{BBeta}(\boldsymbol{\theta}_\delta)$, $\rho= 0$, $\alpha>1,\beta>1$, $\delta>1$ and
\begin{eqnarray}\label{condition-zero}
	\big[\delta(\alpha+1)+\alpha+\beta-2\big]^2>4\delta(\alpha+\beta)(\alpha-1),
\end{eqnarray}	
	then the \text{BBeta} distribution is bimodal.
\end{thm1}
\begin{proof}
When $\rho= 0$, in \eqref{derivative}, we have
{\scalefont{0.97}
\begin{align}\label{flinha}
f'(x;\boldsymbol{\theta}_\delta)
&=
{x^{\alpha-2}(1-x)^{\beta-2} (1-\delta x)\over Z(\boldsymbol{\theta}_\delta) {B}(\alpha,\beta)} \,
\Big\{
(1-\delta x)\big[(\alpha-1)(1-x)-(\beta-1)x\big]
-2\delta x(1-x)\Big\}
\\[0,1cm]
&=
{x^{\alpha-2}(1-x)^{\beta-2} (1-\delta x)\over Z(\boldsymbol{\theta}_\delta) {B}(\alpha,\beta)} \,
\Big\{
x^2(\alpha+\beta)\delta-\big[\delta(\alpha+1)+\alpha+\beta-2\big]x+(\alpha-1)
\Big\}. \nonumber
\end{align}
}
A direct calculus shows that  $f'(x;\boldsymbol{\theta}_\delta)=0$ if and only if (excluding the boundary points) $x=1/\delta$ and
\begin{align*}
x_{\pm}=
{
	\delta(\alpha+1)+\alpha+\beta-2 \pm
	\sqrt{\big[\delta(\alpha+1)+\alpha+\beta-2\big]^2-4\delta(\alpha+\beta)(\alpha-1)}
\over
2\delta(\alpha+\beta)
}.
\end{align*}

Note that, by conditions $\alpha>1,\beta>1$, $\delta>1$, in \eqref{flinha} we have $f'(x;\boldsymbol{\theta}_\delta)\neq 0$ for all $x>1$.
%
Hence, under condition \eqref{condition-zero}, it follows that the equation $f'(x;\boldsymbol{\theta}_\delta)=0$ has three positive roots $x=1/\delta,x_-$ and $x_+$ within the interval $(0,1)$, where $x_-<x=1/\delta<x_+$.

Since, for $\alpha>1,\beta>1$, $f(x; \boldsymbol{\theta}_\delta) \longrightarrow 0$ as $x \to 0^+$
and $f(x; \boldsymbol{\theta}_\delta) \longrightarrow 0$  as
$x\to 1^-$, the bimodality of the BBeta distribution is guaranteed, where $x_-$ and $x_+$ are two
maximum points and $x=1/\delta$ is the unique minimum point.
\end{proof}

\vspace*{-0,5cm}
\section{Some characteristics and properties}
\label{sect:3}
\noindent

In this section, some closed expressions for the mean residual life function and real moments of the BBeta distribution are obtained.

\begin{thm1}\label{truncated-moments}
	If $X\sim \text{BBeta}(\boldsymbol{\theta}_\delta)$ then, for $0\leqslant a<b\leqslant 1$ and $r> -\alpha$,
	\begin{eqnarray*}
	\mathbb{E}\big(X^r\mathds{1}_{\{a\leqslant X\leqslant b\}}\big)
	=
	\dfrac{1}{ Z(\boldsymbol{\theta}_\delta)  }
	\sum_{i=0}^{2}
	c_i\,
	\biggl[
 \dfrac{B_{b}(\alpha+r+i, \beta)}{B(\alpha, \beta)}
	-
	\dfrac{B_{a}(\alpha+r+i, \beta)}{B(\alpha, \beta)}
	\biggr],
	\end{eqnarray*}
	where $c_0=1+\rho$, $c_1=-2\delta$, $c_2=\delta^2$, and
	$B_x(\alpha, \beta)$ is the incomplete beta function.
\end{thm1}
\begin{proof}
By using definition of expectation and definition of BBeta density, we have
	\begin{eqnarray*}
			\mathbb{E}\big(X^r \mathds{1}_{\{a\leqslant X\leqslant b\}}\big)
				=
		\frac{1}{ Z(\boldsymbol{\theta}_\delta) }
			\sum_{i=0}^{2}
		c_i\, \mathbb{E}\big(Y^{r+i} \mathds{1}_{\{a\leqslant Y\leqslant b\}}\big), \quad Y\sim \text{BBeta}(\boldsymbol{\theta}_0).
	\end{eqnarray*}
Since
\begin{eqnarray*}
\mathbb{E}\big(Y^{r+i} \mathds{1}_{\{a\leqslant Y\leqslant b\}}\big)
=
\dfrac{B_{b}(\alpha+r+i, \beta)}{B(\alpha, \beta)}
-
\dfrac{B_{a}(\alpha+r+i, \beta)}{B(\alpha, \beta)},
\end{eqnarray*}
the proof of theorem follows.
\end{proof}

Taking $r = 0$, $b=x$ and $a = 0$ in Theorem \ref{truncated-moments}, we get the formula \eqref{CDF} for the CDF.
%
Letting $r = 0$, $b=1$ and $a = x$ in Theorem \ref{truncated-moments}, we get the formula \eqref{SF} for the SF.

\begin{corollary}[Mean residual life function]
	If $X\sim \text{BBeta}(\boldsymbol{\theta}_\delta)$ then mean residual life function of $X$, defined by ${\rm MRL}(x,\boldsymbol{\theta}_\delta)=\int_{x}^{1} S(t;\boldsymbol{\theta}_\delta)\, {\rm d}t/S(x;\boldsymbol{\theta}_\delta)$, is written as
	{
	\begin{eqnarray*}
	{\rm MRL}(x,\boldsymbol{\theta}_\delta)=
	\dfrac{	\sum_{i=0}^{2}
		c_i\,
		\big\{
		\big[B(\alpha+i+1, \beta)-x B(\alpha+i, \beta)\big]
		-
		\big[B_{x}(\alpha+i+1, \beta)-x B_{x}(\alpha+i, \beta)\big]
		\big\}}{\sum_{i=0}^{2}
		c_i\,
		\big[
		B(\alpha+i, \beta)
		-
		B_{x}(\alpha+i, \beta)\big]},
	\end{eqnarray*}
}
	where $c_0=1+\rho$, $c_1=-2\delta$ and $c_2=\delta^2$.
\end{corollary}
\begin{proof}
Integration by parts gives
\begin{eqnarray}\label{MRL}
{\rm MRL}(x,\boldsymbol{\theta}_\delta)={1\over S(x,\boldsymbol{\theta}_\delta)} \, \mathbb{E}\big(X\mathds{1}_{\{X\geqslant x\}}\big) -x.
\end{eqnarray}
Taking $r=1$, $a=x$ and $b=1$ in Theorem \ref{truncated-moments}, we get
	\begin{eqnarray*}
	\mathbb{E}\big(X\mathds{1}_{\{X\geqslant x\}}\big)
	=
	\dfrac{1}{ Z(\boldsymbol{\theta}_\delta)  }
	\sum_{i=0}^{2}
	c_i\,
	\biggl[
	\dfrac{B(\alpha+i+1, \beta)}{B(\alpha, \beta)}
	-
	\dfrac{B_{x}(\alpha+i+1, \beta)}{B(\alpha, \beta)}
	\biggr].
\end{eqnarray*}
By replacing the above identity in \eqref{MRL}, the proof follows.
\end{proof}

By combining the formula \eqref{SF} of CDF and definition of the BBeta distribution, we obtain the formula \eqref{HR} for the HR.

\begin{corollary}[Real moments]\label{Real moments}
If $X\sim \text{BBeta}(\boldsymbol{\theta}_\delta)$ and $r> -\alpha$, then
\begin{eqnarray*}
\mathbb{E}(X^r)
	=
	\dfrac{1}{ Z(\boldsymbol{\theta}_\delta) }
	\biggl[
	(1+\rho)\,
	\dfrac{B(\alpha+r, \beta)}{B(\alpha, \beta)}
	-2\delta\,
	\dfrac{B(\alpha+r+1, \beta)}{B(\alpha, \beta)}
	+\delta^2\,
	\dfrac{B(\alpha+r+2, \beta)}{B(\alpha, \beta)}
	\biggr],
\end{eqnarray*}
\end{corollary}
\begin{proof}
By taking  $b=1$ and $a = 0$ in Theorem \ref{truncated-moments} we have the following:	
\begin{eqnarray*}
	\mathbb{E}(X^r)
	=
	\dfrac{1}{ Z(\boldsymbol{\theta}_\delta)  }
	\sum_{i=0}^{2}
	c_i\,
	\dfrac{B(\alpha+r+i, \beta)}{B(\alpha, \beta)},
\end{eqnarray*}	
where $c_0=1+\rho$, $c_1=-2\delta$ and $c_2=\delta^2$.
\end{proof}


\begin{corollary}[Raw moments]\label{moments}
	If $X\sim \text{BBeta}(\boldsymbol{\theta}_\delta)$ and $k\in[0,+\infty)\cap\mathbb{Z}$, then
	\begin{align*}
		\mathbb{E}(X^k)
		=
		\dfrac{1}{Z(\boldsymbol{\theta}_\delta)}
		    \Biggl(\prod_{j=0}^{k-1}\frac{\alpha+j}{\alpha+\beta+j}\Biggr)
		   \biggl[
		   1+\rho	
		   -
		   {2\delta(\alpha+k)\over \alpha+\beta+k}
		   +
		  {\delta^2(\alpha+k)(\alpha+k+1)\over (\alpha+\beta+k)(\alpha+\beta+k+1)}
		   \biggl],
	\end{align*}
where we are conventioning that $\prod_{j=0}^{-1}({\alpha+j})/({\alpha+\beta+j})=1$.
\end{corollary}
\begin{proof}
 By taking $r=k$ in Corollary \ref{Real moments} and using the simple recurrence relation
\begin{eqnarray}\label{rel-rec}
	B(x+k,y)=B(x,y) \, \prod_{j=0}^{k-1}\frac{x+j}{x+y+j}
\end{eqnarray}
we have
	\begin{eqnarray*}
	\mathbb{E}(X^k)
	=
	\dfrac{1}{ Z(\boldsymbol{\theta}_\delta)  }
	\sum_{i=0}^{2}
	c_i
\prod_{j=0}^{k+i-1}\frac{\alpha+j}{\alpha+\beta+j},
\end{eqnarray*}	
where $c_0=1+\rho$, $c_1=-2\delta$ and $c_2=\delta^2$.
From the above formula the proof follows immediately.
\end{proof}
As a consequence of the above corollary, the closed expressions for the standardized moments, variance, skewness and kurtosis of the bimodal Beta r.v. $X$ are easily obtained.

\begin{rem1}\label{obs-1}
	Taking $\delta=0$ in in Corollaries \ref{Real moments} and \ref{moments}  we obtain the following known formulas:
\begin{align*}
	\mathbb{E}(Y^r)=
	\dfrac{B(\alpha+t, \beta)}{B(\alpha, \beta)}, \ r>-\alpha;
	\quad 
	\mathbb{E}(Y^k)
	=
	\prod_{j=0}^{k-1}\frac{\alpha+j}{\alpha+\beta+j}, \
	k\in[0,+\infty)\cap\mathbb{Z}; \quad Y\sim \text{BBeta}(\boldsymbol{\theta}_0).
\end{align*}	
\end{rem1}

An immediate application of Corollary \ref{moments} provides the following result.
\begin{corollary}
	If $X\sim \text{BBeta}(\boldsymbol{\theta}_\delta)$ then
\begin{align*}
\mathbb{E}(X)&=
\dfrac{1}{Z(\boldsymbol{\theta}_\delta)} \,
\frac{ \alpha}{\alpha+\beta}\,
\biggl[
1+\rho		   -
{2\delta(\alpha+1)\over \alpha+\beta+1}
+
{\delta^2(\alpha+1)(\alpha+2)\over (\alpha+\beta+1)(\alpha+\beta+2)}
\biggl];
\\[0,2cm]
\mathbb{E}(X^2)&=
\dfrac{1}{Z(\boldsymbol{\theta}_\delta)}\,
\frac{\alpha(\alpha+1)}{(\alpha+\beta)(\alpha+\beta+1)}\,
\biggl[
1+\rho	   -
{2\delta(\alpha+2)\over \alpha+\beta+2}
+
{\delta^2(\alpha+2)(\alpha+3)\over (\alpha+\beta+2)(\alpha+\beta+3)}
\biggl].
\end{align*}
\end{corollary}

\begin{rem1}
The deformed moment generating function of BBeta r.v. $X$ is given by the following expression:
\begin{equation*}
\mathbb{E}\big[\exp_q(tX)\big]
=	
{\Gamma(\alpha)\Gamma(\beta) (A_1 + A_2) \over {B}(\alpha,\beta)\Gamma(\alpha+\beta) Z(\boldsymbol{\theta}_{\delta})},
\quad q \in [0,1), t \geqslant 0,
\end{equation*}
\noindent
where
$\exp_q(tx)=[1 + (1 - q)tx]^{1/ (1-q)}$ denotes the deformed exponential function,
$A_1=(1+\rho){H}_2({1/(q-1)},\alpha,\alpha+\beta,t(q-1))$ and
{\small
\begin{align*}
A_2={
-2(1+\alpha+\beta) {H}_2({1 \over q-1},1+\alpha,1+\alpha+\beta,t(q-1))
+
(1+\alpha) \delta {H}_2({1 \over q-1},2+\alpha,2+\alpha+\beta,t(q-1))
\over
(\alpha + \beta )(1 + \alpha + \beta)/(\delta \alpha)}.
\end{align*}
}
Here,
$H_2(a,b,c,z)$ is the hypergeometric function $_2 F_1 (a,b;c;z)$. By using L'Hospital's rule we have that, if $q\to 1$, $\exp_q(tx)$ drops to $\exp(tx)$.
\end{rem1}

\begin{corollary}[Moment generating function]\label{momentMGF}
	If $X\sim \text{BBeta}(\boldsymbol{\theta}_\delta)$ and $t \geqslant 0$, then
	{\small
	\begin{align*}
	\mathbb{E}\big[\exp(tX)\big]
	&=\Gamma(\alpha + \beta){ t^{-1}
	\biggl\{1+\rho +
	\frac{\alpha \delta [\delta(1+\alpha)-2(1+\alpha+\beta)]}{(\alpha+\beta)(1+\alpha+\beta)}\biggr\}^{-1}}
	\\[0,2cm] \nonumber
	&
	\times\left\{
	t[(\delta-1)^2 + \rho]
	-\beta \delta^2 H_1(\alpha,\alpha+\beta,t)
	+
	 \beta \delta [\delta(\alpha+\beta)+t(2-\delta)] H_1(\alpha,1+\alpha+\beta,t)
	\right\},
\end{align*}
}
where
$H_1(a,b,c)$ is the regularized confluent hypergeometric function $_1 F_1(a;b;z)/\Gamma(b)$.

\end{corollary}

\vspace*{-0,5cm}
\section{Further properties}
\label{sect:4}
\noindent

In this section, we consider some properties of the BBeta distribution, such as the
entropy measures, stochastic representation and identifiability.

\subsection{Entropy measures}
\noindent

Let $X\sim \text{BBeta}(\boldsymbol{\theta}_\delta)$.
The Tsallis \citep{Tsallis1988} entropy associated with a non-negative random variable $X$ is defined by
\begin{eqnarray*}
S_q(X)=\dfrac{1}{q-1}\, \bigg[1-\int_{0}^{1} f^q(x;\boldsymbol{\theta}_\delta) \, {\rm d}x\bigg], \quad q\neq 1.
\end{eqnarray*}

The quadratic entropy \citep{Rao2010} is defined as
\begin{eqnarray*}
H_2(X)=-\log \int_{0}^1 f^2(x;\boldsymbol{\theta}_\delta) \, {\rm d}x.
\end{eqnarray*}

We also define the Shannon entropy \citep{Shannon1948} as
\begin{eqnarray*}
H_1(X)=
-
\int_{0}^{1}
f(x;\boldsymbol{\theta}_\delta)
\log f(x;\boldsymbol{\theta}_\delta) \, {\rm d}x.
\end{eqnarray*}

By using L'Hospital's Rule, we have that, if $q\to 1$, then $S_q(X) \to H_1(X)$ and the usual definition of Shannon's entropy is recovered.

	\begin{thm1}[Tsallis entropy] 
	Let $X\sim \text{BBeta}(\boldsymbol{\theta}_\delta)$, $\alpha_q=q(\alpha-1)+1>0$, $\beta_q=q(\beta-1)+1>0$, $\rho\geqslant 1$ and $0\leqslant q<1$. Then
		\begin{align*}
		\int_{0}^{1} f^q(x;\boldsymbol{\theta}_\delta) \, {\rm d}x
		\leqslant
\frac{ {B}(\alpha_q,\beta_q)}{ [Z(\boldsymbol{\theta}_\delta) {B}(\alpha,\beta)]^q } \,
\biggl[
1+q\rho
+
q\delta^2\,
{{B}(\alpha_q+2,\beta_q)\over {B}(\alpha_q,\beta_q)}
-
2q\delta\,
{{B}(\alpha_q+1,\beta_q)\over {B}(\alpha_q,\beta_q)}
\biggr].
		\end{align*}
	Moreover, the two sides are equal if and only if $q$ is sufficiently close to 1.
	
	In particular, for $\alpha_q>0$, $\beta_q>0$ and $0\leqslant q<1$, the Tsallis entropy  exists.	
\end{thm1}
\begin{proof}
	By definition of BBeta PDF, we have
\begin{eqnarray}\label{exp-1}
\int_{0}^{1} f^q(x;\boldsymbol{\theta}_\delta) \, {\rm d}x
=
\frac{1}{ [Z(\boldsymbol{\theta}_\delta) {B}(\alpha,\beta)]^q } \,
\int_{0}^{1}
[\rho+(1-\delta{x})^2]^q
x^{q(\alpha-1)} \, (1-x)^{q(\beta-1)}
\, {\rm d}x.
\end{eqnarray}
By using the inequality (see, e.g., \cite{Hardy34})
$
a^b\leqslant 1+(a-1)b, \ \text{for} \ b\in[0,1], \ a\geqslant 1,
$
the expression on the right-hand side of \eqref{exp-1} is at most
\begin{align*}
&\frac{1}{ [Z(\boldsymbol{\theta}_\delta) {B}(\alpha,\beta)]^q } \,
\int_{0}^{1}
\big\{1+[\rho-1+(1-\delta x)^2] q\big\}\, x^{q(\alpha-1)} \, (1-x)^{q(\beta-1)}
\, {\rm d}x
\\[0,2cm]
&=
\frac{{B}(q(\alpha-1)+1,q(\beta-1)+1)}{ [Z(\boldsymbol{\theta}_\delta) {B}(\alpha,\beta)]^q } \,
[1+q\rho+q\delta^2\mathbb{E}(Y^2)-2q\delta\mathbb{E}(Y)],
\end{align*}
where $Y\sim \text{BBeta}(q(\alpha-1)+1,q(\beta-1)+1)$.

By applying Remark \ref{obs-1}, the proof follows.	
\end{proof}

By using the identifiability (see Subsection \ref{Identifiability}), it is possible to write an upper bound
	for the Tsallis entropy and log$_q(f)$, where, for $x > 0$, $\log_q(x)=(x^{1-q}-1)/(1-q)$, $q\neq 1$, represents the deformed logarithm \citep{Tsallisbook09}. After having an upper bound for the Tsallis entropy and log$_q(f)$,
	the application of MLqE method for the estimation of the parameters of BBeta will be accurate. 
	
	If the proposed distribution can have entropies, existence of MGF (see Corollary \ref{momentMGF}), it is safe to apply for modelling on a data set. Otherwise, unboundness and nonexistence of moments for PDF cannot provide modelling many types of real data sets and free from computational error which can occur while performing optimization in order to get the estimators of parameters in the distribution \cite{gut13}.

\begin{proposition}[Quadratic entropy]
Let $X\sim \text{BBeta}(\boldsymbol{\theta}_\delta)$ with $\alpha>1/2$, $\beta>1/2$. Then
\begin{align*}
H_2(X)&=
-\log B(2\alpha-1,2\beta-1)
+
2\log Z(\boldsymbol{\theta}_\delta)
+
2\log B(\alpha,\beta)
-
\log\Biggl[
\sum_{i=0}^4
\widetilde{c}_i \prod_{j=0}^{i-1}\frac{\alpha+j}{\alpha+\beta+j}
\Biggr],
\end{align*}
where $\widetilde{c}_0=(1 + \rho)^2$, $\widetilde{c}_1=- 4 \delta(1+\rho)$, $\widetilde{c}_2=2 \delta^2(3+\rho)$, $\widetilde{c}_3=- 4 \delta^3$ and $\widetilde{c}_4=\delta^4$.	
\end{proposition}
\begin{proof}
Since $\alpha>1/2$ and $\beta>1/2$, by using definitions of density $f$ and expectation,
\begin{align}\label{exp-1-1}
\int_{0}^{1} f^2(x;\boldsymbol{\theta}_\delta) \, {\rm d}x
&=
\frac{{B}(2\alpha-1,2\beta-1)}{[Z(\boldsymbol{\theta}_\delta){B}(\alpha,\beta)]^2} \,
\int_{0}^{1}
[\rho+(1-\delta{x})^2]^2\,
{x^{2(\alpha-1)} \, (1-x)^{2(\beta-1)}\over {B}(2\alpha-1,2\beta-1)}
\, {\rm d}x \nonumber
\\[0,3cm]
&=
\frac{{B}(2\alpha-1,2\beta-1)}{[Z(\boldsymbol{\theta}_\delta){B}(\alpha,\beta)]^2} \,
\mathbb{E}\big\{[\rho+(1-\delta{Y})^2]^2\big\},
\quad Y\sim{\rm Beta}(2\alpha-1,2\beta-1).
\end{align}
Developing the quadratic factor above,
\begin{align*}
\mathbb{E}\big\{[\rho+(1-\delta{Y})^2]^2\big\}
=
(1 + \rho)^2
- 4 \delta(1+\rho) \mathbb{E}(Y)
+ 2 \delta^2(3+\rho) \mathbb{E}(Y^2)
- 4 \delta^3 \mathbb{E}(Y^3) + \delta^4 \mathbb{E}(Y^4)
\end{align*}
and replacing in \eqref{exp-1-1}, we have
\begin{align*}
\int_{0}^{1} f^2(x;\boldsymbol{\theta}_\delta) \, {\rm d}x
=
\frac{{B}(2\alpha-1,2\beta-1)}{[Z(\boldsymbol{\theta}_\delta){B}(\alpha,\beta)]^2} \,
\sum_{i=0}^4
\widetilde{c}_i \mathbb{E}(Y^i),
\end{align*}
with $\widetilde{c}_0=(1 + \rho)^2$, $\widetilde{c}_1=- 4 \delta(1+\rho)$, $\widetilde{c}_2=2 \delta^2(3+\rho)$, $\widetilde{c}_3=- 4 \delta^3$ and $\widetilde{c}_4=\delta^4$.	

Hence, from Remark \ref{obs-1} and definition of quadratic entropy, the proof follows.
\end{proof}

\begin{lemma}[The $1$-th logarithmic moment about zero]\label{exp-log}
If $X\sim \text{BBeta}(\boldsymbol{\theta}_\delta)$, with $\alpha\geqslant 2$, then
\begin{eqnarray*}
	\mathbb{E}\big[\log (X)\big]
	=
	\frac{1}{ Z(\boldsymbol{\theta}_\delta) B(\alpha, \beta)}
	\sum_{i=0}^{2}
	c_i \,
	\dfrac{\partial B(\alpha+i, \beta)}{\partial\alpha},
\end{eqnarray*}
where $c_0=1+\rho$, $c_1=-2\delta$ and $c_2=\delta^2$, and $\psi(x)=\Gamma'(x)/\Gamma(x)$ is the digamma function.
\end{lemma}
\begin{proof} By using definition of expectation of a function of a BBeta r.v. $X$, we have
\begin{align*}
		\mathbb{E}\big[\log (X)\big]=
			\frac{1}{ Z(\boldsymbol{\theta}_\delta) }
		\sum_{i=0}^{2}
		c_i\, \mathbb{E}\big[Y^{i} \log (Y)\big], \quad Y\sim \text{BBeta}(\boldsymbol{\theta}_0).
\end{align*}

If we prove that
\begin{eqnarray}\label{claim-1}
\mathbb{E}\big[Y^{i} \log(Y)\big]
=
\dfrac{1}{B(\alpha, \beta)}\,
\dfrac{\partial B(\alpha+i, \beta)}{\partial\alpha}, \quad i=0,1,2,\ldots,
\end{eqnarray}
the proof follows. In what remains of the proof, we show the validity of \eqref{claim-1}.

Indeed, since ${\partial y^{\alpha-1}/\partial \alpha}=(\log y) y^{\alpha-1}$, we get
\begin{eqnarray}\label{id-integral}
	\mathbb{E}\big[Y^{i} \log(Y)\big]
	&=&
	\dfrac{1}{B(\alpha, \beta)}\,
	\int_{0}^{1} \log(y)\, {y^{\alpha+i-1}(1 - y)^{\beta-1}} \, {\rm d} y \nonumber
	\\[0,15cm]
	&=&
	\dfrac{1}{B(\alpha, \beta)}\,
	\int_{0}^{1}
	\dfrac{\partial}{\partial\alpha}\big[
	{y^{\alpha+i-1}(1 - y)^{\beta-1}}\big] \, {\rm d} y.
\end{eqnarray}
%
%
A standard calculation shows that conditions of Leibniz integral rule are satisfied. Then we can interchange the derivative with the integral in \eqref{id-integral}. Hence
\begin{align*}
	\mathbb{E}\big[Y^{i} \log(Y)\big]
	&=
	\dfrac{1}{B(\alpha, \beta)}\,
	\dfrac{\partial}{\partial\alpha}
	\int_{0}^{1}  {y^{\alpha+i-1}(1 - y)^{\beta-1}} \, {\rm d} y	
=
	\dfrac{1}{B(\alpha, \beta)}\,
	\dfrac{\partial B(\alpha+i,\beta)}{\partial\alpha},
\end{align*}
and \eqref{claim-1} follows.

Thus, we complete the proof of lemma.
\end{proof}

\begin{rem1}\label{rem-main}
By using Lemma \ref{exp-log},  the identity
$
{\partial \log B(\alpha, \beta)}/{\partial \alpha}=\psi(\alpha)-\psi(\alpha+\beta)
$,
with
$\psi(x)=\Gamma'(x)/\Gamma(x)$,
and the recurrence relation \eqref{rel-rec}, we have
\begin{align*}
	\mathbb{E}\big[\log (X)\big]
	=
	\frac{1}{ Z(\boldsymbol{\theta}_\delta) }
	\sum_{i=0}^{2}
		c_i \,
	\biggl[
 \big(\psi(\alpha)-\psi(\alpha+\beta)\big) \prod_{j=0}^{i-1}\frac{\alpha+j}{\alpha+\beta+j}
	+
	{\partial\over \partial\alpha}
\prod_{j=0}^{i-1}\frac{\alpha+j}{\alpha+\beta+j}
\biggr].
\end{align*}
\end{rem1}

\begin{thm1}[Shannon entropy]
Let $X\sim \text{BBeta}(\boldsymbol{\theta}_\delta)$, with $\rho=0$, $\alpha\geqslant 2$ and $\delta=1$. Then
\begin{multline*}
\hspace*{-0.35cm}
H_1(X)=	  \log \Gamma(\alpha)+\log\Gamma(\beta) -\log\Gamma(\alpha + \beta)
+
\log\biggl[1
-2\, {\alpha\over \alpha+\beta}
+ {\alpha(\alpha+1)\over(\alpha+\beta)(\alpha+\beta+1)}\biggr]
\\[0,2cm]
 - 	\frac{(\alpha-1)\beta(\beta+1)}{ (\alpha+\beta)(\alpha+\beta+1)
 	-2\alpha(\alpha+\beta+1)
 	+
 	 {\alpha(\alpha+1)} }
\left[\psi(\alpha)-\psi(\alpha+\beta)-{2(\alpha+\beta)+1\over   (\alpha+\beta) (\alpha+\beta+1)}\right]
 \\[0,2cm]
 + \dfrac{(\alpha+\beta)(\alpha+\beta+1)(\beta+1)}{ (\alpha+\beta)(\alpha+\beta+1)
 	-2\alpha(\alpha+\beta+1)
 	+
  {\alpha(\alpha+1)}  }
 \sum_{i=0}^{2}
 c_i
 \sum_{k=1}^{\infty}
 \dfrac{1}{k}
 \prod_{j=0}^{k+i-1}\frac{\alpha+j}{\alpha+\beta+j},
\end{multline*}
whenever the series above converges absolutely.
Here, $c_0=c_2=1$ and $c_1=-2$, and $\psi(x)=\Gamma'(x)/\Gamma(x)$ is the digamma function.
\end{thm1}
\begin{proof}
	Since $\rho=0$ and $\delta=1$,
	a simple computation shows that
	\begin{multline}\label{Step-1}
	\int_{0}^{1}
	f(x;\boldsymbol{\theta}_1)
	\log f(x;\boldsymbol{\theta}_1) \, {\rm d}x	
	=
\mathbb{E}\big[\log f(X;\boldsymbol{\theta}_1)\big]
\\[0,2cm]
=
	  -\log Z(\boldsymbol{\theta}_1) -\log{B}(\alpha,\beta)
	 + (\alpha-1)\,\mathbb{E}\big[\log (X)\big] + (\beta+1)\,\mathbb{E}\big[\log(1-X)\big].
\end{multline}

Taking $c_0=c_2=1$ and $c_1=-2$ in Remark \ref{rem-main} we obtain
\begin{align}\label{Step-2}
	\mathbb{E}\big[\log (X)\big]
	=
	\frac{\beta(\beta+1)}{ Z(\boldsymbol{\theta}_1)\, (\alpha+\beta) (\alpha+\beta+1)}
\left[\psi(\alpha)-\psi(\alpha+\beta)-{2(\alpha+\beta)+1\over   (\alpha+\beta) (\alpha+\beta+1)}\right].
\end{align}

In what follows we provide a closed expression for the expectation
$\mathbb{E}\big[\log(1-X)\big]$.
Indeed, by using series representation of function $\log(1-x)$; also called Newton-Mercator series:
$
\log(1-x)
=
-
\sum_{k=1}^{\infty}
{x^k}/{k}
$
which converges for $0<x<1$, we have
\begin{align}\label{Step-3}
\mathbb{E}\big[\log(1-X)\big]
=
-\sum_{k=1}^{\infty}
\dfrac{\mathbb{E}(X^k)}{k}
=
-\dfrac{1}{ Z(\boldsymbol{\theta}_1)  }
\sum_{i=0}^{2}
c_i
\sum_{k=1}^{\infty}
\dfrac{1}{k}
\prod_{j=0}^{k+i-1}\frac{\alpha+j}{\alpha+\beta+j},
\end{align}
where in the second equality we used Corollary \ref{moments}.

By combining \eqref{Step-1}, \eqref{Step-2} and \eqref{Step-3}, and using definitions of normalization constant $Z(\boldsymbol{\theta}_1)$ and beta function ${B}(\alpha,\beta)$, the proof follows.
\end{proof}

\subsection{Stochastic representation}
\noindent

We say that a r.v. $Y$ has a non-standard Beta distribution bounded in $[a,b]$ interval and shape parameters $\alpha > 0$ and $\beta > 0$,
if its PDF is given by
\begin{equation*}
g(x; \alpha, \beta,a,b) = \frac{(x-a)^{\alpha-1}( b-x)^{\beta-1}}{B(\alpha, \beta) (b-a)^{\alpha+\beta-1}}, \quad a \leqslant {x} \leqslant b.
\end{equation*}

\begin{proposition}[Stochastic representation for $\delta<0$]\label{Stochastic representation}
	Suppose $Y_{k;\alpha,\beta}$ has a non-standard Beta distribution bounded in $[0,1/k]$ interval, with $k=1,2,3$, and shape parameters $\alpha > 0$ and $\beta > 0$. Let $W$ be a discrete distribution, so that $W=1$ or $W=2$ or $W=3$, each with probability
\begin{align*}
\pi_1={1+\rho\over Z(\boldsymbol{\theta}_\delta)}, \quad  \pi_2=-{2\alpha \delta\over Z(\boldsymbol{\theta}_\delta)\, (\alpha+\beta)},
\quad
\pi_3={\alpha(\alpha+1)\delta^2 \over Z(\boldsymbol{\theta}_\delta)\, (\alpha+\beta)(\alpha+\beta+1)},
\end{align*}
respectively, with $\delta<0$. A simple algebraic manipulation shows that $\pi_1+\pi_2+\pi_3=1$.

Assume that
$$Y
=
\sum_{k=1}^{3}Y_{k;\alpha+k-1,\beta} \delta_{W,k} \delta_{W,l},
\quad l=1,2,3,
$$
and that $W$ is independent of $Y_{k;\alpha,\beta}$, for each $k=1,2,3$. Here $\delta_{W,k}$ is the Kronecker delta function, i.e., $\delta_{W(\omega),k}$ is 1 if $W(\omega)=k$ for $\omega$ belonging to the random sample $\Omega$, and 0 otherwise.
	
	If $X=WY$
	then $X\sim \text{BBeta}(\boldsymbol{\theta}_\delta)$.	
	Conversely, if $X\sim \text{BBeta}(\boldsymbol{\theta}_\delta)$ then $X=WY$.
\end{proposition}
\begin{proof}
By Law of total probability and by independence, we get
\begin{align*}
\mathbb{P}(X\leqslant x)
=
\mathbb{P}(WY\leqslant x)
&=
\sum_{l=1}^{3}
\mathbb{P}(WY\leqslant x\vert W=l) \mathbb{P}(W=l)
\\[0,15cm]
&=
\sum_{l=1}^{3}
\mathbb{P}(lY_{l;\alpha+l-1,\beta}\leqslant x)\mathbb{P}(W=l)
=
\sum_{l=1}^{3}
\mathbb{P}(Y_{1;\alpha+l-1,\beta}\leqslant x)\pi_l,
\end{align*}
because $lY_{l;\alpha+l-1,\beta}=Y_{1;\alpha+l-1,\beta}\sim {B}(\alpha+l-1,\beta)$, for $l=1,2,3$.
Since for $Y \sim {B}(\alpha,\beta)$ its CDF is given by
$
F(x; \alpha, \beta) = I_{x}(\alpha, \beta),
\ 0 \leqslant {x} \leqslant 1,
$
by definition of $\pi_l$'s, the above expression is
\begin{align*}
=
\sum_{l=1}^{3}
I_{x}(\alpha+l-1, \beta)\pi_l
=
\dfrac{1}{ Z(\boldsymbol{\theta}_\delta) }
\biggl[
(1+\rho)\,
{I_{x}(\alpha, \beta)}
-2\delta\,
\dfrac{B_{x}(\alpha+1, \beta)}{B(\alpha, \beta)}
+\delta^2\,
\dfrac{B_{x}(\alpha+2, \beta)}{B(\alpha, \beta)}
\biggr]
.
\end{align*}
But, by \eqref{CDF}, the right-hand side is equal to the CDF $F(x;\boldsymbol{\theta}_\delta)$.

Then we have completed the proof.
\end{proof}

\subsection{Identifiability}\label{Identifiability}
\noindent

Let us suppose that $f(x; \alpha, \beta)$ is the PDF of the Beta distribution, where $\alpha>0$ and
$\beta>0$ are the shape parameters.
A simple observation shows that the bimodal Beta PDF $f(x;\boldsymbol{\theta}_\delta)$ in \eqref{beta-density}, with parameter vector $\boldsymbol{\theta}_\delta =(\alpha,\beta,\rho,\delta)$, can be written as a finite (generalized) mixture of three Beta distributions with different shape parameters, i.e.
\begin{align}\label{eq-density}
f(x;\boldsymbol{\theta}_\delta)
=
\pi_1 f(x;\alpha,\beta)
+
\pi_2 f(x;\alpha+1,\beta)
+
\pi_3 f(x;\alpha+2,\beta),
\quad 0\leqslant x\leqslant 1,
\end{align}
where $\pi_1$, $\pi_2$ and $\pi_3$ are constants (that depends only on $\boldsymbol{\theta}_\delta$) given in Proposition \ref{Stochastic representation},
and $Z({\boldsymbol{\theta}_\delta})$
is as in \eqref{partition-function}.
Unlike Proposition \ref{Stochastic representation}, here $\delta$ can be negative. In principle, mixing non-negative weights are not necessary since mixtures can be PDF even if some of weights are negative.

\smallskip
Let $\mathcal{B}$ be the family of Beta distributions, as follows:
\begin{align*}
\mathcal{B}=\biggl\{F:F(x;\alpha,\beta)=\int_{0}^{x}f(y;\alpha,\beta)\, {\rm d}y, \ \alpha>0, \beta>0,\  0\leqslant x\leqslant 1 \biggl\}.
\end{align*}

Write $\mathcal{H}_{\mathcal{B}}$ the class of all finite mixtures of $\mathcal{B}$. It is well-known that the class $\mathcal{H}_{\mathcal{B}}$ is identifiable (this fact is a consequence of the main result of \cite{Atienza06}).

The following result proves  the identifiability of bimodal Beta distribution.
\begin{proposition}
	The mapping $\boldsymbol{\theta}_\delta \longmapsto f(\cdot;\boldsymbol{\theta}_\delta)$ is one-to-one.
\end{proposition}
\begin{proof}
	Let us suposse that $f(x;\boldsymbol{\theta}_\delta)=f(x;\boldsymbol{\theta}'_\delta)$ for all $0\leqslant x\leqslant 1$. In other words, by \eqref{eq-density},
	\begin{multline*}
	\pi_1 f(x;\alpha,\beta)+
	\pi_2 f(x;\alpha+1,\beta)+
	\pi_3 f(x;\alpha+2,\beta)
	\\[0,2cm]
	=
	\pi_1' f(x;\alpha',\beta')+
	\pi_2' f(x;\alpha'+1,\beta')+
	\pi_3' f(x;\alpha'+2,\beta').
	\end{multline*}
	Since $\mathcal{H}_{\mathcal{B}}$ is identifiable, we have $\pi_i=\pi_i'$, for $i=1,2,3$, and $\alpha=\alpha'$, $\beta=\beta'$. Hence, from equalities $\pi_i=\pi_i'$, $i=1,2,3$, immediately follows that $\rho=\rho'$ and $\delta=\delta'$. Therefore, $\boldsymbol{\theta}_\delta=\boldsymbol{\theta}'_\delta$, and the proof follows.
\end{proof}

\vspace*{-0,5cm}
\section{Regression model, estimation and diagnostic analysis}
\label{sect:5}
\noindent

Let $X_1, \ldots, X_n$ be $n$ independent random variables, where each $X_i$, $i = 1, \ldots, n$, follows the PDF given in~\eqref{beta-density}.
We assume that the parameters $\alpha_i$ and $\beta_i$ satisfy the following functional relations:
\begin{equation}\label{cs1}
g_1(\alpha_i) = \eta_{1i} = \mathbf{w}^\top_i\bm{\gamma} \quad \textrm{and} \quad g_2(\beta_i) = \eta_{2i} = \mathbf{z}^\top_i\bm{\zeta},
\end{equation}
where $\bm{\gamma} = (\gamma_1, \ldots, \gamma_p)^\top$ and $\bm{\zeta} = (\zeta_1, \ldots, \zeta_q)^\top$
are vectors of unknown regression coefficients which are assumed to be functionally independent,
$\bm{\gamma} \in \mathbb{R}^p$ and $\bm{\zeta} \in \mathbb{R}^q$, with $p + q < n$,
$\eta_{1i}$ and $\eta_{2i}$ are the linear predictors, and $\mathbf{w}_i = (w_{i1}, \ldots, w_{ip})^\top$
 and $\mathbf{z}_i = (z_{i1}, \ldots, z_{iq})^\top$ are observations on $p$ and $q$ known regressors, for $i = 1, \ldots, n$. Furthermore, we assume that the covariate matrices $\mathbf{W} = (\mathbf{w}_1, \ldots, \mathbf{w}_n)^\top$ and $\mathbf{Z} = (\mathbf{z}_1, \ldots, \mathbf{z}_n)^\top$ have rank $p$ and $q$, respectively. The link functions $g_1: \mathbb{R} \rightarrow \mathbb{R}^+$ and $g_2: \mathbb{R} \rightarrow \mathbb{R}^+$ in (\ref{cs1}) must be strictly monotone, positive and at least twice differentiable,
such that $\alpha_i = g_1^{-1}(\mathbf{x}_i^\top\,\bm{\gamma})$ and $\beta_i = g_2^{-1}(\mathbf{z}_i^\top\,\bm{\zeta})$, with $g_1^{-1}(\cdot)$ and
$g_2^{-1}(\cdot)$ being the inverse functions of $g_1(\cdot)$ and $g_2(\cdot)$, respectively.

The log-likelihood function for $\boldsymbol{\theta}_\delta = (\bm{\gamma} , \bm{\zeta}, \rho, \delta)$ based on a sample of $ n $ independent
observations is given by
\begin{equation}\label{eq:logm}
\ell(\boldsymbol{\theta}_\delta) = \sum_{i=1}^{n}\ell(\alpha_i, \beta_i, \rho, \delta),
\end{equation}
where
\begin{align*}
\ell(\alpha_i, \beta_i, \rho, \delta)
&=
-\log Z(\boldsymbol{\theta}_\delta) -\log {B}(\alpha_i, \beta_i)
+
\log\big[\rho+(1-\delta{x_i})^2\big]\\
&+
(\alpha_i-1)\log x_i  +(\beta_i-1)\log(1-x_i),
\quad i=1,\ldots,n,
\end{align*}
and $Z(\boldsymbol{\theta}_\delta)$ is as in \eqref{partition-function}.

The maximum likelihood estimator (MLE)
$\widehat{\bm{\theta}}_\delta = (\widehat{\bm{\gamma}}^\top, \widehat{\bm{\zeta}}^\top, \widehat{\rho}, \widehat{\delta})^\top$  of
$\bm{\theta}_\delta  = (\bm{\gamma}^\top, \bm{\zeta}^\top, \rho, \delta)^\top$
is obtained by the maximization of the log-likelihood function \eqref{eq:logm}.
However, it is not possible to derive analytical solution for the MLE $\widehat{\bm{\theta}}$, hence
we must be required to numerical solution using some optimization algorithm such as Newton-Raphson and quasi-Newton.

Under mild regularity conditions and when $ n $ is large,
the asymptotic distribution of the MLE
$\widehat{\bm{\theta}}_\delta = (\widehat{\bm{\gamma}}^\top, \widehat{\bm{\zeta}}^\top, \widehat{\rho}, \widehat{\delta})^\top$
is approximately multivariate normal (of dimension $p+q+2$) with mean vector
$\bm{\theta}_\delta  = (\bm{\gamma}^\top, \bm{\zeta}^\top, \rho, \delta)^\top$  and variance covariance matrix
$\mathbf{K}^{-1}(\bm{\theta}_\delta )$ where
$$\mathbf{K}(\bm{\theta}_\delta )= \mathbb{E}\left[- \ {\partial \ell \left(\bm{\theta}_\delta \right)\over \partial \bm{\theta}_\delta \; \partial \bm{\theta}_\delta^\top} \right],$$
is the  expected Fisher information matrix.
Unfortunately, there is no closed form expression for the matrix $\mathbf{K}(\bm{\theta}_\delta )$.
Nevertheless, a consistent estimator of the expected Fisher information matrix is given by
$$\mathbf{J}(\widehat{\bm{\theta}}_\delta)=- \ {\partial \ell \left(\bm{\theta}_\delta \right)\over \partial \bm{\theta}_\delta \;\partial \bm{\theta}_\delta^\top} \Big{|}_{\bm{\theta}_\delta = \widehat{\bm{\theta}}_\delta} \ ,$$
which is the estimated observed Fisher information matrix.
Therefore, for large $n$, we can replace $\mathbf{K}(\bm{\theta}_\delta )$ by $\mathbf{J}(\widehat{\bm{\theta}}_\delta )$.

Let $ \theta_{\delta_r}$ be the \emph{r}-th component of $\bm{\theta}_\delta .$
The asymptotic $ 100 (1 - \varphi)\% $ confidence interval for $ \theta_{\delta_r} $ is given by
\begin{equation*}
\widehat{\theta}_r \pm z_{\varphi/2}\; \textrm{se}\left(\widehat{\theta}_{\delta_r}\right), \qquad          r = 1, \ldots, p+q+2,
\end{equation*}
where $ z_{\varphi/2}$ is the $ \varphi/2$ upper quantile of the standard normal distribution and
$ \textrm{se}\left(\widehat{\theta}_{\delta_r}\right) $ is the asymptotic standard error of $ \widehat{\theta}_{\delta_r}.$
Note that $ \textrm{se}\left(\widehat{\theta}_{\delta_r}\right) $ is  the square root of the \emph{r}-th diagonal element of the matrix $\mathbf{J}^{-1}(\widehat{\bm{\theta}}_\delta)$.

Residuals are widely used to check the adequacy of the fitted model.
To check the goodness of fit of the BBeta model,
we propose to use the randomized quantile residuals introduced by \cite{du96}.
Let $F(x_i;\boldsymbol{\theta}_\delta)$ be the cumulative distribution
function of the BBeta distribution, as defined in (\ref{CDF}), in which the regression
structures are assumed as in (\ref{cs1}). The randomized quantile residual is given by
$$r_i = \Phi^{-1}\left(F(x_i;\boldsymbol{\widehat\theta}_\delta)\right), \quad i = 1, \ldots, n,$$
where $\Phi^{-1}(\cdot)$ is the standard normal distribution function.
If the assumed model for the data is well adjusted, these residuals have standard
normal distribution \citep{du96}.

\vspace*{-0,5cm}
\section{Simulation study}
\label{sect:6}
\noindent

In this section, Monte Carlo simulations are performed
(i)  to evaluate the finite-sample behavior of the maximum likelihood estimates of the regression coefficients and
(ii) to investigate the empirical distribution of the randomized quantile residuals.

The Monte Carlo experiments were carried out by considering the following regression structure
\begin{align*}
\log\left(\alpha_i\right) &= \gamma_0 + \gamma_1\,z_i,
\nonumber
\\ \nonumber
\log\left(\beta_i\right) &= \zeta_0 + \zeta_1\,z_i, \quad i = 1, \ldots, n, \nonumber
\end{align*}
where the true values of the parameters were chosen to be same with the values of the estimated parameters for the case in which we use the application part of regression, i.e.,
$\gamma_0 = -1.8, \gamma_1 = 5.9, \zeta_0 = 3.8,
\zeta_1 = -2.4, \rho = 0.1$ and $\delta = 2.4$.
The covariate values of $ z_i $ were generated from the standard uniform
distribution.
The sample size considered was $n = 50, 100, 200$ and $300$.
All simulations were conducted in \textsf{R} using the
BFGS algorithm available in the \texttt{optim} function.
For each scenario the Monte Carlo experiment was repeated $5,000$ times.

\subsection{Parameter estimation}
\noindent

In this subsection, a small simulation study is presented to observe the
finite sample performance of the proposed estimators from regression approach.
For such evaluation, the estimated relative bias and the estimated mean squared error (RMSE) were calculated.
The results are presented in Table \ref{tab:parms_sim} and Figure \ref{fig:bps}.

Table \ref{tab:parms_sim} presents the bias and RMSE for the maximum likelihood estimators of
$\gamma_0, \gamma_1, \zeta_0, \zeta_1, \rho$ and $\delta$.
Based on these tables, we find that the estimates are convergent to
their values. As expected, increasing the sample size reduces substantially both
bias and RMSE. The previous findings are confirmed by the box plots shown in Figure \ref{fig:bps}.

\begin{table}[H]
\caption{Estimated bias and mean-squared error.}
\label{tab:parms_sim}
\onehalfspacing
    \scalefont{0.8}
    \centering
    \begin{tabular}{crrrrrrrrrrrr} \midrule
    \multirow{2}{*}{$n$} & \multicolumn{6}{c}{Bias}  & \multicolumn{6}{c}{RMSE} \\
    \cmidrule(lr){2-7}   \cmidrule(lr){8-13}
     & $\gamma_0$ & $\gamma_1$ & $\zeta_0$ & $\zeta_1$ & $\delta$ & $\rho$
    & $\gamma_0$ & $\gamma_1$ & $\zeta_0$ & $\zeta_1$ & $\delta$ & $\rho$ \\
    \midrule   50 & 0.212 & 0.106 & 0.132 & 0.299 & 0.177 & 1.306 & 0.234 & 0.634 & 0.417 & 0.839 & 0.488 & 0.235 \\
100 & 0.213 & 0.099 & 0.114 & 0.254 & 0.120 & 0.938 & 0.192 & 0.475 & 0.276 & 0.558 & 0.183 & 0.091 \\
200 & 0.202 & 0.093 & 0.095 & 0.215 & 0.081 & 0.543 & 0.157 & 0.390 & 0.181 & 0.381 & 0.068 & 0.006 \\
300 & 0.195 & 0.091 & 0.088 & 0.200 & 0.061 & 0.414 & 0.139 & 0.353 & 0.152 & 0.313 & 0.037 & 0.003 \\
\bottomrule \end{tabular}
\end{table}

\begin{figure}[H]
	\centering
	\includegraphics[width=0.9\linewidth]{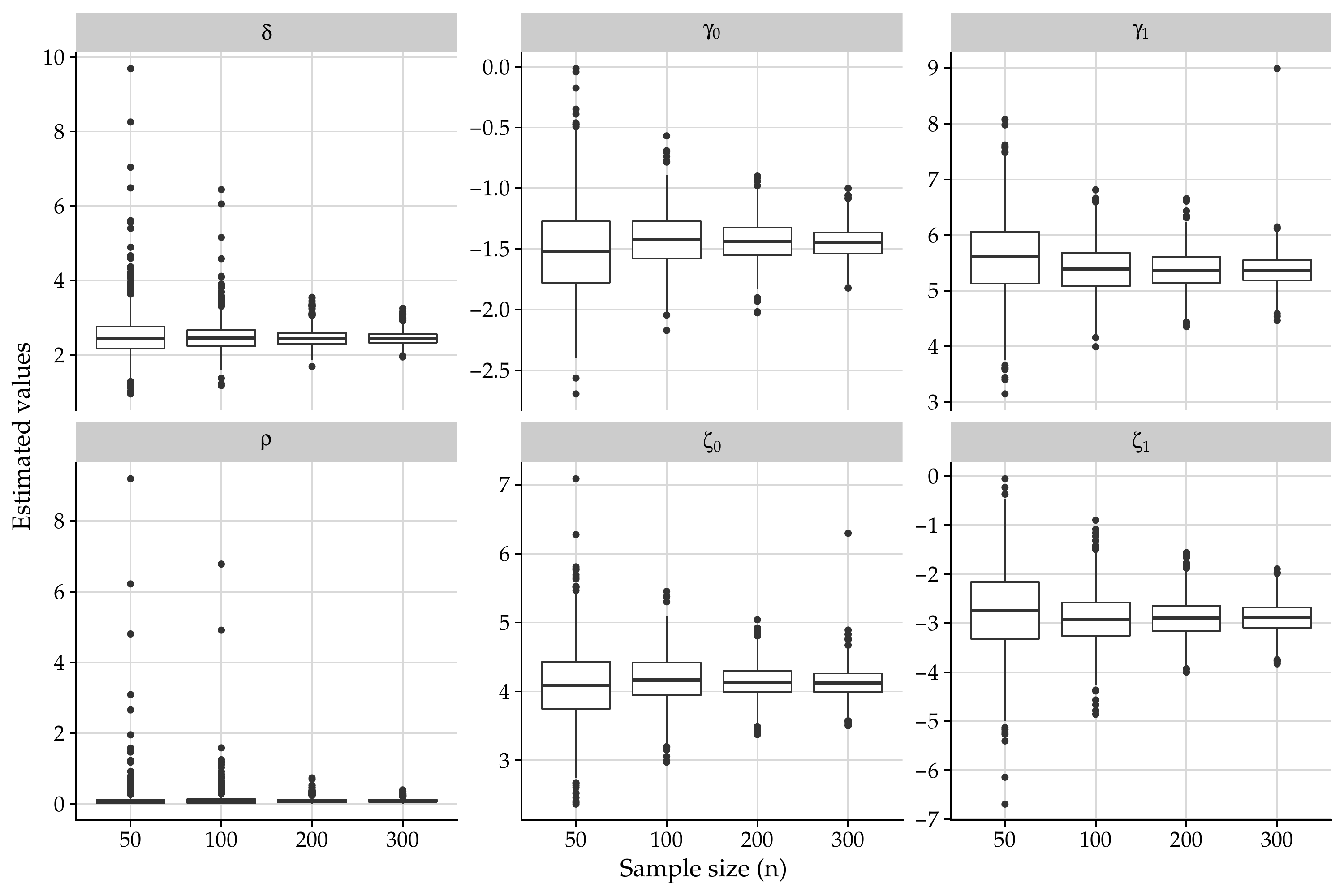}
	\caption{Boxplots of the estimated parameters obtained in Monte Carlo experiments for different sample sizes.}
	\label{fig:bps}
\end{figure}

\subsection{Residuals}
\noindent

The second simulation study was performed to examine how well the distributions of the randomized quantile residuals is approximated by the standard normal distribution. The evaluation of the randomized quantile residuals were based on the normal probability plots of the mean order statistics and descriptive measures. The results are presented in Table \ref{tab:sim_res} and Figure \ref{fig:normalplots}.

In Table \ref{tab:sim_res}, we present the mean, standard deviation (StdDev), skewness and kurtosis of the
randomized quantile residuals. For all scenarios, that is, the residuals have approximately zero mean and unit standard deviation, have skewness close to zero, and the kurtosis is near three.

Figure \ref{fig:normalplots} displays empirical quantiles versus theoretical quantiles plots of the randomized
quantile residuals. The results presented in Figure \ref{fig:normalplots} show that the
distribution of the randomized quantile residuals is approximated by the standard normal distribution.

\begin{table}[H]
\caption{Descriptive measures of the randomized quantile residuals.}
\label{tab:sim_res}
\onehalfspacing
\centering
    \begin{tabular}{crrrr} \midrule
    $n$ & Mean & StdDev & Skewness & Kurtosis \\ \midrule 50 & $-$0.001 & 0.999 & 0.028 & 2.854 \\
100 & $-$0.002 & 0.999 & 0.054 & 2.976 \\
200 & $-$0.003 & 0.997 & 0.077 & 3.002 \\
300 & $-$0.003 & 0.997 & 0.084 & 3.025 \\
\bottomrule \end{tabular}
\end{table}

\begin{figure}[H]
	\centering
	\includegraphics[scale = 0.68]{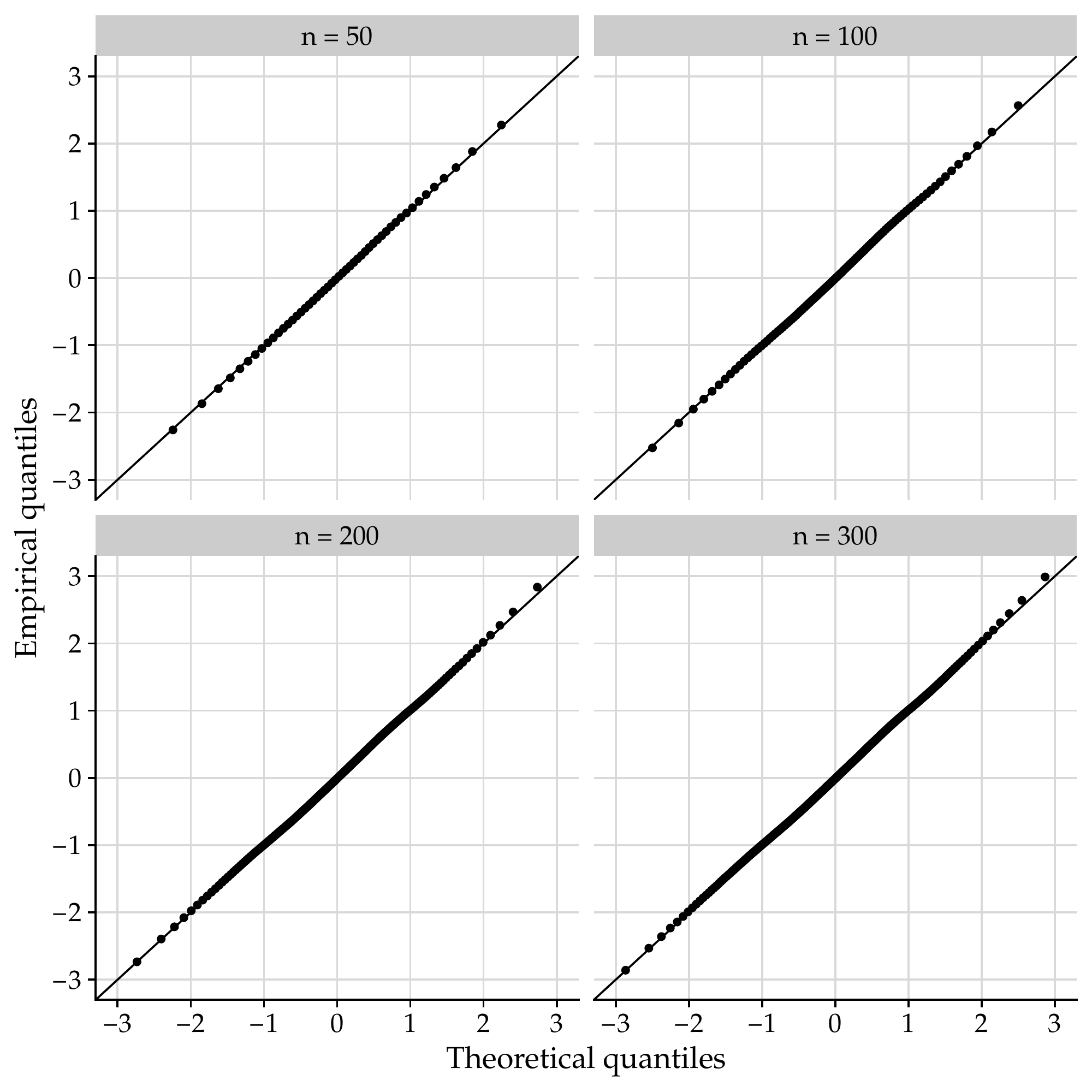}
	\caption{Normal probability plots of the mean order statistics.}
	\label{fig:normalplots}
\end{figure}

\vspace*{-0,5cm}
\section{Real data application}
\label{sect:7}
\noindent

In this section, to evaluate the applicability of the proposed model, a real data set with bimodality is considered.
In particular, a real life application related to the proportion of votes that Jair Bolsonaro received in the second turn of Brazilian elections
in 2018 is analyzed. We compared the potentiality of the BBeta regression with the traditional beta regression model.
In order to estimate the parameters of model, we adopt the MLE method (as discussed in section \ref{sect:5}).
The asymptotic standard errors and confidence intervals were computed using the observed Fisher
information matrix. The required numerical evaluations for data analysis were implemented using the R software.

The goal of this data analysis is to describe the proportion of votes that Jair Bolsonaro received in the second turn of Brazilian elections in 2018 for all 5.565 cities.
The response variable $X_i$ is the proportion of votes given the municipal human development (\textrm{mhdi}).
Figure \ref{empirical_plots:pdf} plots the histogram of response variable used
in the application and the scatterplots of municipal human development against proportion of votes.
From Figure \ref{empirical_plots:pdf}, we can see that the response variable has bimodality.
Furthermore, there is evidence of an proportion of votes trend with increased municipal human development.

\begin{figure}[H]
	\centering
	\setkeys{Gin}{width=0.43\textwidth,height=8.0cm} %
	\includegraphics[]{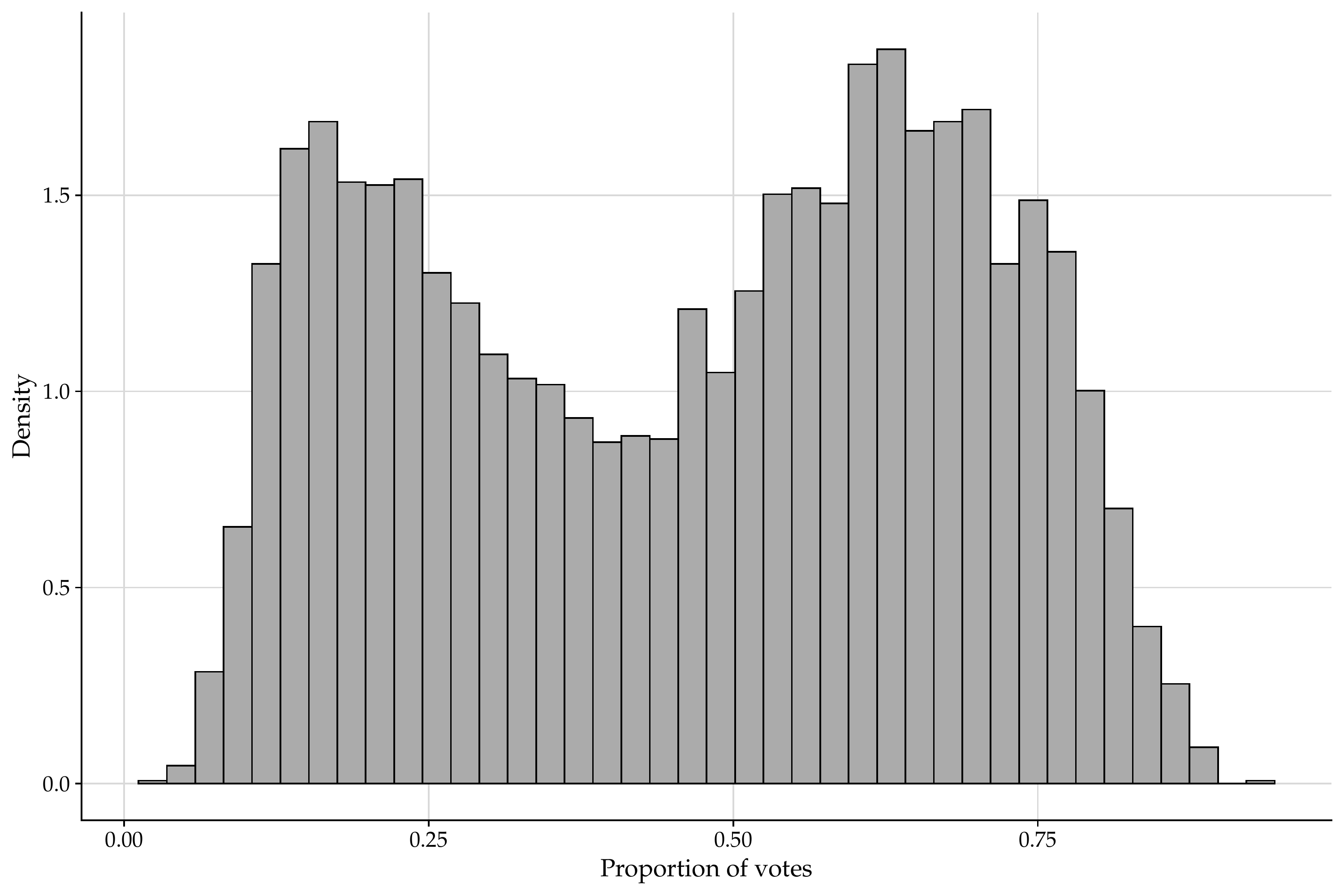}
	\includegraphics[]{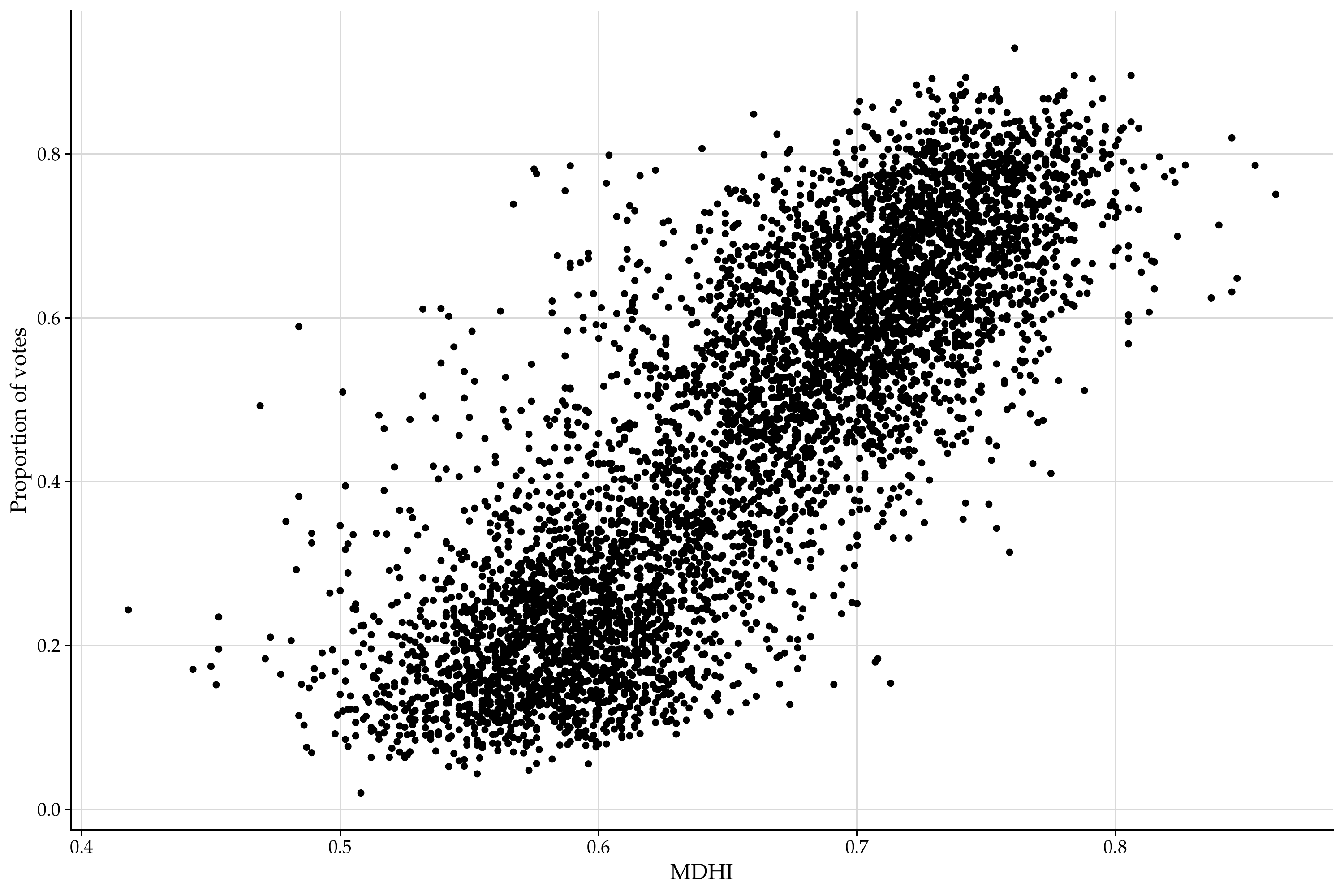}
	\caption{Empirical plots of data.}
	\label{empirical_plots:pdf}
\end{figure}

To explain this proportion of votes we consider the bimodal beta regression model, defined as
\begin{eqnarray*}
Y_i &\sim& \textrm{BBeta}(\bm{\theta}_\delta), \\
\log(\alpha_i) &=& \gamma_0 + \gamma_1\,\textrm{mhdi}_i, \\
\log(\beta_i) &=& \zeta_0 + \zeta_1\,\textrm{mhdi}_i,
\end{eqnarray*}
where $i = 1, 2, \ldots, 5.565$ cities and $\textrm{mhdi}_i$ is municipal human development of cities $i$.
For comparison purposes the beta regression model was fitted, assuming that
\begin{eqnarray*}
Y_i &\sim& \textrm{Beta}(\mu_i, \phi_i), \\
\textrm{logit}(\mu_i) &=& \beta_0 + \beta_1\,\textrm{mhdi}_i, \\
\log(\phi_i) &=& \gamma_0 + \gamma_1\,\textrm{mhdi}_i.
\end{eqnarray*}

Table \ref{estimates} shows the estimated parameters, standard errors and inferior and superior
bounds of the confidence intervals with significance level at 5\% under the BBeta and
Beta models. Note that the coefficients are statistically significant at the
the level of 5\%, for the BBeta and
Beta regression models with the structure above.

\begin{table}[H]
\centering
\caption{ML estimates, standard errors and 95\% confidence interval.}\label{estimates}
\begin{tabular}{ccrrrr}
  \hline
Model & Parameter & Estimate & S.E. & 2.5 \% & 97.5 \% \\
  \hline
\multirow{6}{*}{BBeta} & $\gamma_0$ & $-$1.8999 & 0.1963 & $-$2.2846 & $-$1.5152 \\
 & $\gamma_1$ & 5.9471 & 0.3044 & 5.3505 & 6.5437 \\
 & $\zeta_0$ & 3.8341 & 0.1915 & 3.4587 & 4.2095 \\
 & $\zeta_1$ & $-$2.4232 & 0.2862 & $-$2.9842 & $-$1.8622 \\
 & $\rho$ & 0.1096 & 0.0090 & 0.0920 & 0.1273 \\
 & $\delta$ & 2.4092 & 0.0351 & 2.3405 & 2.4780 \\
   \hdashline
\multirow{4}{*}{Beta }& $\beta_0$ & $-$7.5343 & 0.0749 & $-$7.6810 & $-$7.3875 \\
     & $\beta_1$ & 11.1820 & 0.1105 & 10.9654 & 11.3987 \\
  & $\gamma_0$ & 1.0029 & 0.1675 & 0.6746 & 1.3312 \\
  & $\gamma_1$ & 2.5214 & 0.2528 & 2.0260 & 3.0169 \\
  \hline
\end{tabular}
\end{table}

Table \ref{AIC} shows the Akaike information criterion
(AIC), Bayesian information criterion (BIC) and Kolmogorov-Smirnov (KS) statistic for the fitted models. In general, it
is expected that the better model to fit the data presents the smaller values for these quantities (AIC and BIC).
Based on the AIC and BIC criteria, the model which provides a better fit in this data set is the BBeta regression model.
These claim is also supported by the
residuals plots with simulated envelopes shown in Figure \ref{empirical_plots:pdf}.

\begin{table}[H]
    \centering
    \caption{Goodness-of-fit measures.}\label{AIC}
    \begin{tabular}{lcrr} \hline
        Model & KS & AIC & BIC \\
        \hline
        Beta  & 0.0203 (0.2014) & $-$8238 & $-$8212 \\
        BBeta & 0.0149 (0.5659) & $-$8786 & $-$8746 \\
    \hline
    \end{tabular}
\end{table}

\begin{figure}[H]
	\centering
	\includegraphics[width=0.7\linewidth]{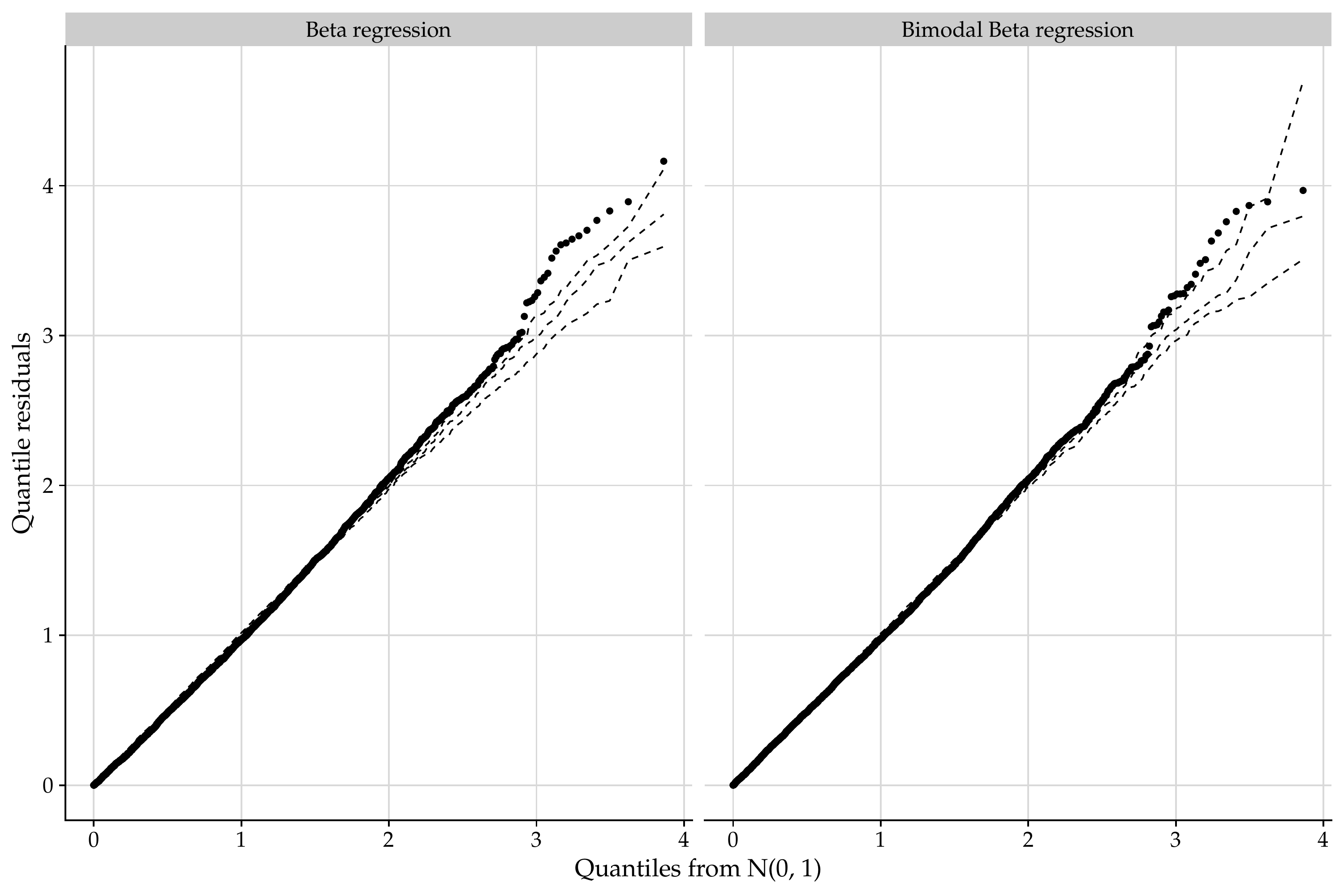}
	\caption{Half-normal plot of randomized quantile residuals with simulated envelope.}
	\label{empirical_plots:pdf}
\end{figure}

\vspace*{-0,05cm}
\section{Concluding remarks}
\label{sect:8}
\noindent

When modeling responses with bimodal bounded to the unit interval, despite its broad sense applicability in many fields, the beta distribution is not suitable.
In this paper, the well-known two-parameter beta distribution is extended by introducing two extra parameters, thus defining the bimodal beta  (BBeta) distribution, based on a quadratic transformation technique used to generate bimodal functions \citep{e:10}, which generalizes the beta distribution. We provide a mathematical treatment of the new
distribution including bimodality, moments, entropy measures, entropy measures, stochastic representation and identifiability.
We allow a regression structure for the parameters $\alpha$ and $\beta$.
The estimation of the model parameters is approached by
maximum likelihood and its good performance has been evaluated
by means of Monte Carlo simulations. Furthermore, we have proposed residuals for the proposed model and conducted a simulation study to establish
their empirical properties in order to evaluate their performances.
The proposed model was fitted to the proportion of votes that Jair Bolsonaro received in the second turn of Brazilian elections in 2018.
As expected, the BBeta model outperforms the beta regression in presence of bimodality.



\begin{thebibliography}{99}
	
\bibitem[Atienza et al.(2006)]{Atienza06} Atienza, N., Garcia-Heras, J., Muñoz-Pichardo, J. M. (2006).
A new condition for identifiability of finite mixture distributions. Metrika, 63, 215--221.



\bibitem[Bayes et al.(2012)]{bayes12}
Bayes, C.L, Baz\'an, J.L, Catalina, G. (2012). A new robust regression model
for proportions. Bayesian Analysis, 7, 841--866.






\bibitem[Dunn and Smyth(1996)]{du96}
Dunn, P.K. and Smyth, G.K., 1996. Randomized quantile residuals. Journal of Computational and
Graphical Statistics, 5, 236--244.



\bibitem[Elal-Olivero(2010)]{e:10}
Elal-Olivero, D. (2010).
Alpha-skew-normal distribution.
Proyecciones Journal of Mathematics, 29:224--240.



\bibitem[Ferrari and Cribari-Neto (2004)]{ferrari2004}
Ferrari, S. and Cribari-Neto, F. (2004).
Beta regression for modelling rates and proportions. Journal of Applied Statistics, 31, 799--815.



\bibitem[Griffiths(1947)]{griffiths1947introduction}
Griffiths, L. (1947).
Introduction to the Theory of Equations. J. Wiley.



\bibitem[Gut(2013)]{gut13}
Gut, A. (2013).
Probability: a graduate course.
(Vol. 75). Springer Science $\&$ Business Media.



\bibitem[Hahn(2021)]{Hahn21}
Hahn, E.D. (2021). Regression modelling with the tilted beta distribution: A Bayesian approach.
The Canadian Journal of Statistics, 49, 262--282.




\bibitem[Hardy et al.(1934)]{Hardy34}
Hardy, G. H., Littlewood, J. E., Pólya, G. (1934).
Inequalities.
Cambridge University Press, Cambridge.



\bibitem[Ji et al.(2005)]{ji05}
Ji, Y. Wu, C., Liu, P., Wang, J., Coombes, K.R. (2005). Applications of beta-mixture models
in bioinformatics. Bioinformatics, 9, 2118--2122.




\bibitem[Johnson et al.(1995)]{jo95}
Johnson, N.L., Kotz, S., Balakrishnan, N. (1995) Continuous Univariate Distributions., vol 2, 2nd edn. John Wiley \& Sons Inc., New York.



\bibitem[Lin et al.(2007a)]{li1}
Lin, T.I., Lee, J.C., Hsieh, W.J. (2007a). Robust mixture models using the skew-t distribution. Statistics and
Computing, 17, 81--92.



\bibitem[Lin et al.(2007b)]{li2}
Lin, T.I., Lee, J.C., Yen, S.Y. (2007b). Finite mixture modeling using the skew-normal distribution. Statistica
Sinica, 17, 909--927.



\bibitem[Ma and Leijon(2009)]{ma09}
Ma, Z., Leijon, A. (2009). Beta mixture models and the application to image classification. Proceedings of
IEEE International Conference on Image Processing (ICIP), 2045--2048.






\bibitem[Olmos et al.(2017)]{olmos17}
Olmos, N.M., Mart\'inez-Fl\'orez, G., Bolfarine, H. (2017). Bimodal Birnbaum-Saunders distribution with
applications to non-negative measurements. Communications in Statistics - Theory and Methods, 46, 6240--6257.



\bibitem[Ospina and Ferrari(2008)]{os08}
Ospina, R., Ferrari, S.L.P. (2008). Inflated beta distributions. Statistical Papers, 51, 111--126.



\bibitem[Rao(2010)]{Rao2010}
Rao, C. R. (2010).
Quadratic entropy and analysis of diversity.
Sankhya A, 72, 70--80.




\bibitem[Shannon(1948)]{Shannon1948}
Shannon, C. E. (1948).
A mathematical theory of communication,
Bell System Technical Journal, 27, 379--423, 623-656.



\bibitem[Smithson and Segale(2009)]{Smithson1}
Smithson, M., Segale, C. (2009). Partition Priming in Judgments of Imprecise Probabilities.
Journal of Statistical Theory and Practice, 3, 169--181.



\bibitem[Smithson et al.(2011)]{Smithson2}
Smithson, M., Merkle, E.C., Verkuilen J. (2011). Beta Regression Finite Mixture Models of
Polarization and Priming." Journal of Educational and Behavioral Statistics, 36, 804--831.



\bibitem[Tsallis(1988)]{Tsallis1988}
Tsallis,  C. (1988).
Possible generalization of Boltzmann-Gibbs statistics.
Journal of Statistical Physics, 52, 479--487.




\bibitem[Tsallis(2009)]{Tsallisbook09}
Tsallis, C.  (2009).
Introduction to Nonextensive Statistical Mechanics: Approaching a Complex World.
Springer, New York.




\bibitem[Vila et al.(2020)]{VFSPO20} 	
Vila, R., Ferreira, L., Saulo, H., Prataviera, F., Ortega, E. M. M. (2020).
A bimodal gamma distribution: Properties, regression model and applications.
Statistics,
54, 469--493.



\bibitem[Vila and \c{C}ankaya(2021)]{VC20} 	
Vila, R., \c{C}ankaya, M. N. (2021).
A Bimodal Weibull Distribution: Properties and Inference.
Journal of Applied Statistics, 1--19.



\bibitem[Vinberg(2003)]{vinberg2003course}
Vinberg, {\.E}. (2003).
A Course in Algebra.
Graduate studies in mathematics. American Mathematical Society.

\bibitem[Wong(2013)]{wong} Wong, M. C. (2013).
Bubble value at risk: A countercyclical risk management approach.
John Wiley \& Sons.

%
\bibitem[Xue(2012)]{xue2012loop}
Xue, J. (2012).
Loop Tiling for Parallelism.
The Springer International Series in Engineering and Computer Science. Springer US.


\end{thebibliography}
\end{document}